\documentclass{article}

\usepackage{microtype}
\usepackage{graphicx}
\usepackage{booktabs} % for professional tables
\usepackage[square,sort,comma,numbers]{natbib}

\usepackage{bbm}
\usepackage{amsfonts}       % blackboard math symbols
\usepackage{amssymb,amsmath,amsthm}
\usepackage{url}
\usepackage[margin=0.5in]{geometry}

\usepackage{algorithm}
\usepackage{algorithmic}
\usepackage{graphicx}
\usepackage{epstopdf}
\usepackage{subfig}

\usepackage{microtype}
\usepackage{graphicx}
\usepackage{booktabs} % for professional tables
\RequirePackage{fancyhdr}
\RequirePackage{color}
\RequirePackage{algorithm}
\RequirePackage{algorithmic}
\RequirePackage{eso-pic} % used by \AddToShipoutPicture 
\RequirePackage{forloop}

\usepackage{epstopdf}
\usepackage{array}

\let\oldnl\nl% Store \nl in \oldnl
\newcommand{\nonl}{\renewcommand{\nl}{\let\nl\oldnl}}% Remove line number for one line

\newtheorem{theorem}{Theorem}
\newtheorem{corollary}{Corollary}

\newtheorem{fact}{Fact}
\newtheorem{remark}{Remark}

\newtheorem{proposition}{Proposition}

\def\DEBUG{true}

\ifdefined\DEBUG

\title{Fair Clustering with Multiple Colors}

\author{%
Matteo B\"ohm, 
Adriano Fazzone,\\ 
Stefano Leonardi,
Chris Schwiegelshohn\\
  Sapienza University of Rome
}

\begin{document}

\maketitle
	
\begin{abstract}
A fair clustering instance is given a data set $A$ in which every point is assigned some color. Colors correspond to various protected attributes such as sex, ethnicity, or age. A fair clustering is an instance where membership  of points in a cluster is uncorrelated with the coloring of the points.

Of particular interest is the case where all colors are equally represented. If we have exactly two colors, Chierrichetti, Kumar, Lattanzi and Vassilvitskii (NIPS 2017) showed that various $k$-clustering objectives admit a constant factor approximation. Since then, a number of follow up work has attempted to extend this result to a multi-color case, though so far, the only known results either result in no-constant factor approximation, apply only to special clustering objectives such as $k$-center, yield bicrititeria approximations,  or require $k$ to be constant. 

In this paper, we present a simple reduction from unconstrained $k$-clustering to fair $k$-clustering for a large range of clustering objectives including $k$-median, $k$-means, and $k$-center. The reduction loses only a constant factor in the approximation guarantee, marking the first true constant factor approximation for many of these problems.
\end{abstract}

\section*{Acknowledgements}
Partially supported by the ERC Advanced Grant 788893 AMDROMA ``Algorithmic and Mechanism Design Research in Online Markets'' and MIUR PRIN project ALGADIMAR ``Algorithms, Games, and Digital Markets''.

\newcommand{\msdk}{\textsc{MSDk}\xspace}

\section{Introduction}
Clustering is one of the fundamental building blocks of data analysis. 
Due to the enormous amount of attention it has received in research,  classic optimization problems 
such as $k$-means or $k$-median are well (if not completely) understood both in theory and practice.
Nevertheless, there exist a number of open questions. In recent years, the researchers have started 
to address clustering with cardinality constraints. If we require clusters to have a maximum or minimum size, 
these problems tend to become far less tractable.

Cardinality constraints arise naturally in a number of settings including but not limited to privacy preserving clustering, capacitated
clustering, and fairness. In this paper we focus on the latter. Given two distinct populations $A$ and $B$ consisting of $n$ points
each, a clustering is considered to be \emph{fair} if $A$ and $B$ are equally represented in every cluster.
Informally speaking, the separation of the point set into two subsets can be a way of modeling specific binary attributes (e.g., sex,
citizenship) against which an algorithm (or indeed a clustering) should not discriminate.
Our aim is to find a good clustering that obeys such a fairness constraint. 

This (comparatively recent) line of research was initiated
by~\cite{CKLV17} and further developed in subsequent
works~\cite{BIOSVW19,BCN19,BGKKRSS19,RS18,SSS19}.
\cite{CKLV17} affirmatively answered the question of
whether a high quality fair clustering can be obtained efficiently for
two distinct populations. Most follow up work (see e.g.~\cite{BCN19,BGKKRSS19,HuangJV19,RS18,SSS19}) has since considered the question of whether these
guarantees can be extended for multiple populations.  Even though some progress has been made for 
specific clustering problems such as $k$-center or by allowing small fairness violations, the general problem  
remains open despite considerable effort by the research community.

\subsection{Our Contribution}
We settle the problem of the design algorithms with good approximations for fair clustering with multiple populations in the affirmative by showing that for all center based clustering objectives (including $k$-means, $k$-median, $k$-center) and all  metrics. Our main result is as follows.

\begin{theorem}[Informal, see Theorem~\ref{thm:approxclustering}]
Given an $\alpha$-approximation for an unconstrained center-based $k$-clustering problem, there exists an $(\alpha+2)$ approximation for the $k$-clustering problem with fairness constraints.
\end{theorem}

Given the large number of good approximation algorithms known for the $k$-clustering problems, our result can be widely applied in most practical setting of interest for fairness applications.  

Previous results either only applied to $k$-center~\cite{BGKKRSS19,RS18}, yielded bicriteria
approximations~\cite{BCN19,BGKKRSS19}, or required $k$  to be constant~\cite{SSS19}.
The general algorithm we propose is quite simple, in contrast to earlier work requiring geometric decompositions~\cite{SSS19,HuangJV19} or rounding linear programs~\cite{BCN19,BGKKRSS19}.  
%The proof of our result is given in its entirety in this work since it provides important guidelines on how to design efficient algorithms that solve this large class of problems.

%At a high level, our techniques are an extension of the  the fairlet approach introduced by Chierichetti et al.~\cite{CKLV17}. Briefly,
%although computing a fair $k$-clustering of two $n$-point sets $A$ and $B$ is hard, Chierichetti et al.~\cite{CKLV17} observed that a fair
%$n$-clustering can be computed via min-cost perfect matching in polynomial time.  However, the approach of Chierichetti et al. cannot be directly extended to multiple attributes since computing a fair $n$-clustering for three or more $n$-point sets (say $A$, $B$, and $C$ for three attributes), is NP-hard (Proposition~\ref{prop:hardness}).  Our approach is based on the result reported in this work showing that a $2$-approximate fair $n$-clustering is efficiently computable even for multiple protected attributes. The main observation needed to prove this result is  that there always exists an $n$-point set that is close to the optimal fairlet.

The caveat of our approach is that it requires solving multiple instances of a transportation problem. We remedy this by giving algorithms for $k$-center and $k$-median that run in linear time. Specifically, we show that for $k$-center, there exists a simple greedy heuristic that induces a
$3$-approximation. For $k$-median, we extend the linear time $O(d\log n)$-approximation by~\cite{BIOSVW19} to multiple
attributes.

Lastly, we also consider hardness of approximation for these problems. If the number of centers is constant and the points lie in Euclidean space, we show that a polynomial time approximation exists. This is complemented with the following hardness proof. Given three point sets consisting of exactly $n$ points, finding a fair $n$-median or $n$-center clustering is APX hard. This already shows that considering fair clustering with at least $3$ populations is harder than the same problem with only $2$ populations.

\subsection{Related work}
\label{sec:related}

\paragraph{Algorithmic Fairness and Fair Clustering}
Fairness in algorithms has recently received considerable attention, see
\cite{HPNS16,TRT11,ZVG17,ZVGGW17c} and references therein.
The idea of clustering using balancing constraints is derived from notion of disparate impact.
Disparate impact was first proposed by~\cite{FFMSV15}. Despite some impossibility results in certain settings~\cite{Chouldechova17,KleinbergMR17}, it has been used successfully combined for classification~\cite{GordalizaBGL19,NBGP18,ZVGG17}, ranking~\cite{CHV18,CSV18}, regression~\cite{AgarwalDW19}, graph embeddings~\cite{BoseH19} and indeed clustering.
Its application to clustering was initiated by~\cite{CKLV17}.
They showed that for two protected classes, fair clustering for various objectives such as $k$-median, $k$-center, and (implicitly, though unstated) $k$-means can be approximated as well as the unconstrained variants of the problem (up to constant factors).
Building upon their work,~\cite{BIOSVW19},~\cite{SSS19}, and~\cite{HuangJV19} considered this problem for large data sets.
The main open problem left in their work is whether the approximability
can be extended for multiple color classes. Here, the $k$-center problem
has received the most attention~\cite{BCN19,BGKKRSS19, RS18}, with the current state of the art being a $5$-approximation or a bicretieria $4$-approximation that violates the fairness constraint by a small amount. 
Conversely, prior to our work, for $k$-means with multiple protected classes, only a PTAS for constant $k$~\cite{SSS19,HuangJV19} in Euclidean spaces of constant dimension and
bicriteria approximation algorithms were known~\cite{BCN19, BGKKRSS19}.

We note that there exist other models combining fairness and clustering objectives. 
Disparity of impact for spectral clustering has been studied by~\cite{KleindessnerSAM19}. Further spectral algorithms with fairness considerations appear in~\cite{MSSTV19, SamadiTMSV18,TantipongpipatS19}.
\cite{KAM19} considered $k$-centers under the fairness constraint that the set of centers, rather than the composition of the clusters. 

\paragraph{Cardinality Constrained Clustering}
As mentioned above, fair clustering is a special case of clustering with cardinality constraints. When the cardinalities are bounded from above, this is known as capacitated clustering. For $k$-center, constant factor approximation algorithms are known~\cite{AnBCGMS15,CyganHK12}. For $k$-median no constant factor approximation is known, though there are several bicriteria approximations~\cite{AdamczykBMM019,DemirciL16,Li17}. 
\cite{APFTKKZ10} and~\cite{RS18} consider a variant of $k$-center where the cluster sizes are lower bounded, which models privacy preserving clustering. For this problem, they obtain a $12$-approximation. 
A similar notion by~\cite{ChenFLM19} phrases fairness in terms of proportionality of a given solution.
\cite{LiYZ10} also considered a $k$-center problem with a diversity constraint. Here, we are given $\ell$ color classes and every cluster must contain at least one point from every class. For this problem, they achieved a $2$-approximation.
For arbitrary constraints on cluster sizes,~\cite{BhattacharyaJK18,DingL18,DingX15} obtain polynomial time approximation schemes, given constant $k$.

\subsection{Preliminaries and Problem Definition}\label{se:prelims}

Throughout this paper, let $[n]$ be the set of natural numbers from $1$ to $n$.
A $k$-clustering of an $n$-point set $A$ is a disjoint partition of $A$ into $k$ subsets $C_1,\ldots, C_k$ called clusters.
Further, we are given a coloring $c:[n] \rightarrow [\ell]$.
A set $S\subset [n]$ is called
\emph{balanced} if $|S\cap \{v\in V~|~c(v)=i\}| = |S\cap \{v\in
V~|~c(v)=j\}|$ for any $i,j\in [l]$. 
A clustering is called {\em balanced} or {\em fair} if every cluster is balanced.
We view our input point set as
%an $n$ by $d$
a matrix $A$, where the rows of $A$ correspond to points and the columns denote features.
If the input point set is balanced, we say that $A$ is an $\ell\cdot n \times d$ matrix and use a coloring $c:[\ell\cdot n]\rightarrow [\ell]$. We further use $A^{(i)}$ to denote the set $\{p\in A~|~c(p)=i\}$.

The $\ell_p$ norm of a $d$-dimensional vector $x$ is denoted $\|x\|_p = \sqrt[p]{\sum_{i=1}^d |x_i|^p}$. Taking the limit of $p\rightarrow \infty$, we define $\|x\|_{\infty} = \max_{i=1}^d x_i$.
For a matrix $A$, define the \emph{cascaded} $(p,q)$ norm \newline
$\|A\|_{p,q} := \left(\sum_{i=1}^n \left(\sum_{j=1}^d |A_{i,j}|^q\right)^{p/q}\right)^{1/p}$,
that is, we first compute the
$q$-norms of the rows of $A$, and then compute the $p$ norm of the
resulting vector. It is perhaps instructive to note that $\|A\|_{2,2}$
is identical to the Frobenius norm of $A$.

The $(k,p,q)$ clustering
problem consists of computing an $n\times d$ matrix $C$ with at most $k$
distinct rows minimizing $\|A-C\|_{p,q}$. 
The cluster induced by the row $C_j$ are the points $\{A_i~|~ i\in [n] \wedge C_i=C_j\}$.
For example, $(k,1,1)$
clustering is the $k$-median problem in Hamming space, and $(k,\infty,2)$ clustering is
Euclidean $k$-center.
The fair-$(k,p,q)$ clustering problem further requires that every cluster induced by $C$ is balanced.

Finally, given two $n$ point sets $A^{(i)}$ and $A^{(j)}$, a matching is a bijection $\pi: A^{(i)}\rightarrow A^{(j)}$. Given some matching $\pi$, we say that the $(p,q)$-cost is $\sqrt[p]{\sum_{x\in A^{(i)}} \|x-\pi(x)\|_q^{p}}$. The optimal matching with respect to the $(p,q)$-cost is called the min-cost perfect $(p,q)$-matching, or simply min-cost perfect matching if $p$ and $q$ are clear from the context. 
In literature, this is sometimes referred to as the Earth-Mover's distance between $A^{(i)}$ and $A^{(j)}$, for which we use the shorthand $EMD_{p,q}(A^{(i)},A^{(j)})$.
The time required to compute an optimal min-cost perfect matching on $n$-point sets is denoted by $MCPM(n)$\footnote{There exist algorithms that run faster in special cases, such as $(1,2)$-matching in low-dimensional Euclidean space. For a single algorithm that solves the problem for all $p$ and $q$, we refer the reader as an example to the Hungarian algorithm~\cite{Kuhn55}.}.

\section{Approximate Fair $k$-Clustering}
%\subsection{Fair $k$-Clustering}
We start with our main result for $(k,p,q)$-clustering objectives such as $k$-median and $k$-center.

\begin{theorem}
\label{thm:approxclustering}
Let $A$ be an $\ell\cdot n\times d$ matrix, let $c:[\ell\cdot n]\rightarrow [\ell]$ be a balanced coloring of $A$, and let $k$ be an integer.
Given an $\alpha$-approximation of the unconstrained
$(k,p,q)$-clustering problem running in time $T(n,d,k,p,q)$, there
exists an $(\alpha + 2)$
approximation for the fair $(k,p,q)$-clustering problem with $\ell$ colors running in time $O(\ell^2\cdot MCPM(n)\cdot T(n,d,k,p,q))$. 
%For fair $k$-means with an arbitrary number of colors, the approximation factor is $2\alpha+4$.
\end{theorem}

\begin{algorithm}
\caption{Fair to Unfair Reduction for $(k,p,q)$-Clustering}
\begin{algorithmic}
\REQUIRE{Balanced point set $A=\biguplus_{i=1}^{\ell} A^{(i)}$ with $|A^{(i)}|=n$, $k\in \mathbb{N}$, $p,q\in \mathbb{N}$\\
}
%\STATE{$OPT\leftarrow (0,\infty)$}
\FOR{$i\in\{1,\ldots,\ell\}$}
%	\STATE{$cost(i)\leftarrow 0$\\}
	\FOR{$j\in\{1,\ldots,\ell\}$}
		\STATE{compute min-cost perfect $(p,q)$ matching $\pi_{i,j}$ between $A^{(i)}$ and $A^{(j)}$\\}
	\ENDFOR
	\STATE{Run unconstrained $(p,q)$-clustering algorithm on $A^{(i)}$ with centers $C^{i}=\{c_1^{i},\ldots c_k^{i}\}$}
\FOR{$j=\{1,\ldots ,\ell\}$}
	\FOR{$x\in A^{(j)}$}
		\STATE{Assign $x$ to $\underset{c\in C}{\text{argmin}} \|\pi_{i,j}(x)-c\|_q$	
}
	\ENDFOR
\ENDFOR
\ENDFOR
\STATE{Output the best fair clustering among all $C^{i}$\\}
\end{algorithmic}
\label{alg:reduction}
\end{algorithm}

We give in Algorithm ~\ref{alg:reduction} the pseudocode of the algorithm used to prove this theorem.
The high level ideas are as follows. Our algorithm is based on the computation of a solution 
for fair $n$-clustering initially proposed by~\cite{CKLV17}. We show that a $2$-approximate fair $n$-clustering can always be obtained
by selecting the input color $A^{(i)}$ such that $\sqrt[p]{\sum_{j=1}^{\ell} \left(EMD_{p,q}(A^{(i)},A^{(j)})\right)^p}$
is minimized. Run an algorithm for unconstrained $(k,p,q)$-clustering on $A^{(i)}$ then allows to recover a fair $k$-clustering of $A$.

What is left to show for the analysis is that clustering (multiplicities of) $A^{(i)}$ can be related to clustering the entire point set. The total cost of the final solution is actually bounded by $2$ times the cost of the  fair $n$-clustering plus the cost of an $\alpha$-approximation algorithm for $(k,p,q)$-clustering.  This allows to recover an $(\alpha+4)$-approximate fair clustering.
To obtain an $(\alpha+2)$-approximation, additional care is needed. Instead of separately bounding $\sqrt[p]{\sum_{j=1}^{\ell}
\left(EMD_{p,q}(A^{(i)},A^{(j)})\right)^p}$ and the cost of clustering $A^{(i)}$, we show that there always exists a color such that the sum of both is small. While the proof is non-constructive, we can recover an appropriate solution by running the approximation algorithm on all $A^{(i)}$, choosing the assignment by pairwise minimization of $EMD_{p,q}(A^{(i)},A^{(j)})$, and finally selecting the cheapest resulting fair clustering.

%Despite not being able to compute it, the existence of fair $n$-clustering is useful. First, it is a natural lower bound on the cost of a fair $k$-clustering. Second, $\ell\cdot n\times d$ matrix $C$ containing the centers of the fair clustering allow us to structure the rows of $A$. Essentially, we can use $C$ as a virtual $n$-point set with color $(\ell+1)$ that lies central to all the other colors $A^{(j)}$, $j\in [\ell]$. Using elementary operations of the rows of $A$ and $C$ then allow us to conclude that there always exist a color $A^{(i)}$ that is a good proxy for $C$.

\begin{proof}
%[Proof of Theorem~\ref{thm:approxclustering}]
%\marginpar{Here we need to provide a guide to the proof.}

%I TRY but the following needs to be rewritten!!

%(indeed, the existence of balanced $1$-clustering would already be sufficient). 
% and similarly we denote by $OPT_n$ the cost of an optimal balanced $(n,p,q)$-clustering. We observe $OPT_n\leq OPT_k$.
We denote by $\mathcal{L}=\{L_1,\ldots ,L_k\}$ the clusters of an optimal balanced $(k,p,q)$-clustering and $\mathcal{C}=\{c_1,\ldots ,c_k\}$ the associated centers. We denote by $OPT_k$ the cost of an optimal balanced $(k,p,q)$-clustering.
We also observe that the existence of balanced $k$-clustering always implies the existence of a balanced $n$-clustering 
%Note that by definition, any fair $n$-clustering has exactly $\ell$ points in each cluster. 
%We denote the point with color $j$ of cluster $L_i$ by $p_{i,j}$. 

For the analysis of the algorithm, it is useful to rearrange the rows of $A$ in a way that they are sorted by cluster membership in the optimal balanced $(k,p,q)$-clustering. First, we partition $A$ into $\ell$ blocks consisting of $n$ rows, where each block represents a color. In each block, the first $|L_1|$ rows are the points in cluster $L_1$ (in arbitrary order), the next $|L_2|$ rows are the points in cluster $L_2$, and so on. 

We now define an auxiliary matrix containing the information on  the centers  of the optimal balanced $(k,p,q)$-clustering.  Let $C(L)$ denote the $n\times d$ matrix such that the row $C(L)_{i}$ contains the center of cluster $L_i$. 
Finally let $C(\mathcal{L})$ be 
$\ell\cdot n\times d$ matrix obtained by  the union of $\ell$ copies of
$C(L)$, that is, $C(\mathcal{L})^T = [ \overbrace{C(L)^T ~ C(L)^T ~\ldots ~ C(L)^T}^{\ell \text{ copies}}]$.

We then have
%\begin{eqnarray*}
$\displaystyle OPT_k  = \|A-C(\mathcal{L})\|_{p,q}$ \newline $\displaystyle=\left(\sum_{i=1}^n \sum_{j=1}^{\ell} 
\left(\|A_{n\cdot(j-1)+i} -C(L_i)\|_q\right)^p\right)^{1/p}. $
%\end{eqnarray*}
Next, denote by $A^{(j)}$ the $n\times d$ matrix whose rows are all 
points of color $j$ (i.e., $A^{(j)}$ corresponds to rows 
$A_{n\cdot(j-1)+1},\ldots , A_{n\cdot j}$ of $A$).
Now consider the  $n$-clustering of minimum cost obtained by using $n$
points of the same color, that is, the matrix minimizing the expression
\begin{eqnarray*}
&
%\min_{i\in [\ell]} 
\left\|\begin{bmatrix}
A^{(1)} \\
A^{(2)} \\
\vdots \\
A^{(\ell)}
\end{bmatrix} - \begin{bmatrix}
A^{(i)} \\
A^{(i)} \\
\vdots \\
A^{(i)}
\end{bmatrix}\right\|_{p,q} \leq 
%\leq 
%\min_{i\in [\ell]}
%\left(
\left\|\begin{bmatrix}
A^{(1)} \\
A^{(2)} \\
\vdots \\
A^{(\ell)}
\end{bmatrix} - \begin{bmatrix}
C(L) \\
C(L) \\
\vdots \\
C(L)
\end{bmatrix}\right\|_{p,q} + 
\left\|\begin{bmatrix}
A^{(i)} \\
A^{(i)} \\
\vdots \\
A^{(i)}
\end{bmatrix} - \begin{bmatrix}
C(L) \\
C(L) \\
\vdots \\
C(L)
\end{bmatrix}\right\|_{p,q}%\right)
,&
\end{eqnarray*}
where inequality follows by application of triangle inequality, for 
every possible value of $i$.
Note that the first summand is always equal to $OPT_k$
and therefore independent of the minimization with respect to $i$.
For the second summand, we observe that the cheapest color $A^{(j)}$ always satisfies 
\begin{equation}
\left\|\begin{bmatrix}
A^{(j)} \\
A^{(j)} \\
\vdots \\
A^{(j)}
\end{bmatrix} - \begin{bmatrix}
C(L) \\
C(L) \\
\vdots \\
C(L)
\end{bmatrix}\right\|_{p,q}\leq OPT_{k}.
\label{eq:thm2_1}
\end{equation}
Therefore, the minimizer $A^{(i)}$ (which is not necessarily equal to
    $A^{(j)}$) wrt the above expression has cost at most
%$2\cdot OPT_n\leq 2\cdot OPT_k$.
$2\cdot OPT_k$.
Moreover, the very same definition of 
the optimal $A^{(i)}$ (left hand side of the inequality above) suggests 
that we can compute $A^{(i)}$ by simply evaluating a 
min-cost perfect $(p,q)$-matching between all pairs $A^{(j)}$ and $A^{(j')}$ for $j,j'\in \ell$ and aggregating these costs by taking the $p$-norm of the pairwise Earth-Mover's distances.  To see this, observe that the matrix 
\begin{eqnarray*}
& \left\|\begin{bmatrix}
A^{(1)} \\
A^{(2)} \\
\vdots \\
A^{(\ell)}
\end{bmatrix} - \begin{bmatrix}
A^{(i)} \\
A^{(i)} \\
\vdots \\
A^{(i)}
\end{bmatrix}\right\|_{p,q} & 
%= \sqrt[p]{\sum_{j=1}^{\ell} \left(EMD_{p,q}(A^{(i)},A^{(j)})\right)^p}
\end{eqnarray*}
induces a matching between $A^{(i)}$ and $A^{(j)}$, for all $j\in [\ell]$, using the bijection that maps $A^{(i)}_t$ to $A^{(j)}_t$. Then
\begin{eqnarray*}
(2\cdot OPT_k)^p &\geq &\sum_{j=1}^{\ell} \sum_{t=1}^n \|A^{(i)}_t - A^{(j)}_t\|_q^p \geq \sum_{j=1}^{\ell} EMD_{p,q}\left(A^{(i)},A^{(j)}\right)^p.
\end{eqnarray*}

Thus we can conclude that the clustering to the color $i$ minimizing the aggregated cost of Earth Mover's distances satisfies
\begin{equation}
\sqrt[p]{\sum_{j=1}^{\ell} EMD_{p,q}\left(A^{(i)},A^{(j)}\right)^p}\leq 2\cdot OPT_{k}.
\label{eq:thm2_2}
\end{equation}

Now, assuming that $A^{(i)}$ satisfies Equations~\ref{eq:thm2_1} and~\ref{eq:thm2_2} (an assumption we will remove in a moment), we run an $\alpha$-approximate algorithm for unconstrained $(k,p,q)$-clustering on $A^{(i)}$. Let $\mathcal{Z}=\{z_1,\ldots ,z_k\}$ be the set of centers computed by the approximation.

The fair assignment is obtained by mapping the points of the remaining
colors $A^{(j)}$ to the center to which their matching partner in
$A^{(i)}$ was assigned, that is, $A^{(j)}_t$ gets mapped to
$\underset{z\in \mathcal{Z}}{\text{argmin}} \|A^{(i)}_t-z\}$. 
Denote the resulting center matrix $Z$, that is, the rows of $Z$ contain the center to which the corresponding rows of $A^{(1)},\ldots, A^{(\ell)}$ were assigned. 
We then have, after suitably rearranging rows of $A$
\begin{eqnarray}
\nonumber
&\left\|\begin{bmatrix}
A^{(1)} \\
A^{(2)} \\
\vdots \\
A^{(\ell)}
\end{bmatrix} - \begin{bmatrix}
Z \\
Z \\
\vdots \\
Z
\end{bmatrix}\right\|_{p,q} 
%\leq &  
\leq 
%\leq \min_{i\in [\ell]}\left(
\left\|\begin{bmatrix}
A^{(1)} \\
A^{(2)} \\
\vdots \\
A^{(\ell)}
\end{bmatrix} - \begin{bmatrix}
A^{(i)} \\
A^{(i)} \\
\vdots \\
A^{(i)}
\end{bmatrix}\right\|_{p,q} + 
\left\|\begin{bmatrix}
A^{(i)} \\
A^{(i)} \\
\vdots \\
A^{(i)}
\end{bmatrix} - \begin{bmatrix}
Z \\
Z \\
\vdots \\
Z
\end{bmatrix}\right\|_{p,q}
%\right),
&\\
&\leq  2\cdot OPT_{k} + \alpha\cdot OPT_k \leq (\alpha + 2)OPT_k&,
\label{eq:thm2_3}
\end{eqnarray}
where the first inequality follows from the triangle inequality and the second inequality holds by the assumption that $A^{(i)}$ satisfies Equations~\ref{eq:thm2_1} and~\ref{eq:thm2_2}.

To remove the assumption on $A^{(i)}$, we simply run an $\alpha$-approximation for all colors and compute the assignment in the aforementioned way, outputting the cheapest one at the end. The cheapest one is guaranteed to satisfy Equation~\ref{eq:thm2_3}.
\end{proof}

\begin{remark}
A minor modification to this proof is necessary when considering the $k$-center objective, i.e., 
the objective is to compute a $(k,\infty,q)$-clustering. In this case, one has to manipulate the norms with some additional care, as some operations such as taking the $p$-th power are no longer well-defined when taking the limit $p\rightarrow \infty$.
However, in many ways the proof becomes easier, while retaining the same line of reasoning.
Since previous papers already published proofs for $k$-center clustering with similar approximation ratios and we will present in the next section a substantially simpler and faster algorithm, we omit details.

We also remark that for fair $k$-means in Euclidean spaces, which corresponds to finding a matrix $C$ with $k$ distinct rows such that $\|A-C\|_{2,2}^2$ is minimized. Using standard manipulations found in \cite{BBCGS19,FSS13}, one can derive an approximation ratio of $(1+\varepsilon)\alpha + 2(1+\frac{1}{\varepsilon})$ for any $\varepsilon>0$ with the analysis of Theorem~\ref{thm:approxclustering}.
\end{remark}

While Theorem~\ref{thm:approxclustering}, as stated, only applies to clustering in an $\ell_q$ space, it immediately generalizes to arbitrary metrics as well, with minor modifications.
This follows from the fact that any finite $n$ point metric can be isometrically embedded into $\|\cdot\|_\infty$ 
with ${O(n\log n)}$ dimensions; see, for instance,
Matousek~\cite{Mat96}. In other words, the problem of computing a fair $k$-clustering in an arbitrary metric space can always be reduced to computing a fair $k$-clustering in $\ell_{\infty}$ space, i.e. fair $(k,p,\infty)$ clustering.
We summarize this in the following corollary.

\begin{corollary}
Suppose we are given a balanced $n$ point set in some finite metric space. Then there exists a polynomial time $(\alpha+2)$ approximation algorithm for fair $k$-clustering.
\end{corollary}

\section{Faster Algorithms for Fair $k$-Median and $k$-Center}
\label{sec:fast}

\paragraph{\bf Fair $k$-Median}

For fair $k$-median, we  obtain an $(\alpha+2)$-approximation as in Theorem~\ref{thm:approxclustering}, albeit in a substantially faster running time. As mentioned in~\cite{BIOSVW19}, computing min-cost perfect matchings is expensive and tends to dominate the running time of fair clustering. In their paper, they proposed an algorithm that computes an $O(d\log n)$-approximate fairlet decomposition for fair $k$-median in nearly linear time\footnote{The dependency on $d$ may be further reduced to $O(\log k)$ using dimension reduction techniques from~\cite{MakarychevMR19}}. 
The result from Theorem~\ref{thm:approxclustering} can be combined with this approach, yielding an $O(d\log n)$-approximation in time $\tilde{O}(\ell^2 d n)$ time. In this paper we briefly illustrate how to obtain a linear time randomized  algorithm (i.e running in time $\tilde{O}(\ell d n)$).

%The algorithm is based on uniform sampling.
%First, we require the following fact that establishes the metric properties of the Earth Mover's distance.
%%Next, we also show that a linear time deterministic algorithm exists that loses a factor $\log \ell$ in the approximation guarantee.
%
%\begin{fact}[Rubner et al.~\cite{RubnerTG00}, Appendix A]
%\label{fact:EMDmetric}
%Let $(X,d)$ be a metric space with points $X$ and distance function $d$. Then the Earth Mover's distance on (weighted) point sets of equal size (or total weight) using $d$ as a ground distance is a metric.
%\end{fact}

Recall that a fairlet, as defined by~\cite{CKLV17} is a $k'$-clustering with possibly more than $k$ centers, for which a single point is used as a representative. Clustering the representatives and merging the fairlets then results in a fair clustering, for any value of $k$. Note that the existence of a fair $k$-clustering always implies the existence of a fair $n$-clustering, for any number of colors. \cite{CKLV17} show that computing an optimal fair $n$-median is possible if we are given only two colors. While the same problem is APX-hard for three colors (see Proposition~\ref{prop:hardness}), the following theorem establishes that a randomly sampled color is always an $2$-approximate fair $n$-median on expectation. Repeating the sampling process allows us to find a good $n$-median clustering with high probability.  The pseudocode is given in Algorithm \ref{alg:randmedian}.

\begin{algorithm}
\caption{Fast Randomized Fair $n$-Median Clustering}
\begin{algorithmic}
\REQUIRE{Balanced point set $A=\biguplus_{i=1}^{\ell} A^{(i)}$ with $|A^{(i)}|=n$\\
}
\FOR{$i\in\{1,\ldots,\log 1/\delta\}$}
\STATE{Sample $A^{(t})\subset A$ uniformly at random}
\STATE{$cost(i)\leftarrow 0$}
\FOR{$j\in\{1,\ldots,\log \ell\}$}
\STATE{Compute (approximate) $EMD(A^{(t)},A^{(j)})$ and add it to cost}
\ENDFOR
\ENDFOR
\STATE{Output $A^{(t)}$ with minimal cost}
\end{algorithmic}
\label{alg:randmedian}
\end{algorithm}

\begin{theorem}
\label{thm:fastfairkmedian}
Let $A$ be an $\ell\cdot n\times d$ matrix, let $c:[\ell\cdot n]\rightarrow [\ell]$ be a balanced coloring of $A$.
Given an algorithm that computes a $\beta$-approximate fair $n$-median clustering with $2$-colors in time $T(n,d)$, there exists a randomized algorithm that computes a $(2\beta)$-approximate fair $n$-median clustering with $\ell$ colors. The algorithm runs in time $O(\ell\cdot T(n,d) \log 1/\delta)$ and succeeds with probability $1-\delta$. 
\end{theorem}
\begin{proof}
We will start by recalling the following fact that establishes the metric properties of the Earth Mover's distance.
\begin{fact}[Rubner et al.~\cite{RubnerTG00}, Appendix A]
\label{fact:EMDmetric}
Let $(X,d)$ be a metric space with points $X$ and distance function $d$. Then the Earth Mover's distance on (weighted) point sets of equal size (or total weight) using $d$ as a ground distance is a metric.
\end{fact}

Given $\ell$ $n$-point sets $A^{(1)},\ldots A^{(\ell)}$ lying in some metric space, the fair $n$-median problem consists of finding an $n$-point set $B$ such that
$$\sum_{i=1}^{\ell} \sum_{j=1}^n \min_{\pi\in \Pi:A^{(i)}\rightarrow B}d(A_{j}^{(i)},\pi(A_{j}^{(i)}) = \sum_{i=1}^{\ell} EMD(A^{(i)},B)$$ is minimized.

We now sample a point set $A^{(t)}$ uniformly at random. Then 
\begin{eqnarray*}
\mathbb{E}[\sum_{i=1}^{\ell} EMD(A^{(t)},A^{(i)})] 
&\leq &\mathbb{E}[\sum_{i=1}^{\ell} EMD(A^{(t)},B) + EMD(A^{(i)},B)] = \sum_{i=1}^{\ell} EMD(A^{(i)},B) + \sum_{i=1}^{\ell}\mathbb{E}[EMD(A^{(t)},B)] \\
&=&  \sum_{i=1}^{\ell} EMD(A^{(i)},B) + \sum_{i=1}^{\ell}\sum_{j=1}^{\ell} \frac{EMD(A^{(j)},B)}{\ell}
= 2\sum_{i=1}^{\ell} EMD(A^{(i)},B),
\end{eqnarray*}
where we use Fact~\ref{fact:EMDmetric} in the inequality.
Hence, a random point set is always a good candidate solution for an approximate fair $n$-median clustering, with probability at least $1/2$. 
Repeating the sampling process $\log 1/\delta$ times and picking the best one yields a $2$-approximation with probability $1-(1-1/2)^{\log(1/\delta)} = 1-\delta$. 

We now run the $\beta$-approximate computation of fair $n$-median with respect to every sampled color $A^{(t)}$. Let $\pi_{i,t}$ be the matching computed by this algorithm, for every $i\in[\ell]$. We then have
\begin{eqnarray*}
\sum_{i=1}^{\ell} \sum_{p\in A^{(i)}}^n d(p,\pi_{i,t}(p)) 
&\leq &\sum_{i=1}^{\ell} \beta \cdot EMD(A^{(i)},A^{(t)}) 
\leq \sum_{i=1}^{\ell} 2\beta \cdot EMD(A^{(i)},B).
\end{eqnarray*}
\end{proof}

\paragraph{\bf Fair $k$-Center}

In the full version of the paper we show that for the special case of $k$-center in finite metrics, we can compute a set of $k$-centers that induce a $3$-approximate fair $k$-clustering. Moreover, this algorithm runs in nearly linear time.  The algorithm is essentially the farthest first traversal that is well known to produce an optimal $2$-approximation for unconstrained metric $k$-center~\cite{Gonzalez85}. 
This result, that improves for fair $k$-center over Theorem \ref{thm:approxclustering}, is based on the following theorem.

\begin{theorem}
Let $A$ be a set of $\ell\times n$ points in a finite metric, let $c:[\ell\cdot n]\rightarrow [\ell]$ be a balanced coloring of $A$, and let $k$ be an integer.
There exists a $O(ndk)$ time algorithm that computes a set of $k$ points $C\subset A$ such that there exists a $3$-approximate fair $k$ clustering using $C$ as centers.
\end{theorem}
\begin{proof}
We argue why the final set of $k$ points $C$ computed by the farthest first heuristic fulfills the desired criteria.

First, consider the case that every point of $C$ is in a different optimal cluster. In this case, we may upper bound the cost of clustering to $C$ by $2OPT$ via the triangle inequality.
If $C$ does not hit all clusters of the optimal clustering, there must be some cluster that is hit at least twice. Let $i$ be the first iteration in which this occurs and denote by $C_{i-1}$ the points collected so far and by $c_i$ the added point. It then holds $d(c_i,C_{i-1})\leq 2OPT$.

By definition of $c_i$,  we know that for any cluster $O_j$ with center $o_j$ not hit by $C_{i-1}$, we have $d(o_j,C_{i-1})<d(c_i,C_{i-1})$. Since the distance of any point $p\in O_j$ to $o_j$ is at most $OPT$, we therefore have $d(p,C_{i-1})\leq d(p,o_j)+d(o_j,C_{i-1}) \leq OPT + d(c_i,C_{i-1}) \leq OPT + 2OPT = 3OPT$.

Finally, we argue why there exists a valid fair clustering with this bound. The union of two disjoint balanced clusters is a balanced cluster. Let $c_j \in O_j \cap C_{i-1}$, for any cluster $O_j$ hit by $C_{i-1}$. We assign all the points of $O_j$ to $c_j$. For any cluster $O_j$ not hit by $C_{i-1}$, we assign the points of $O_j$ to the center minimizing $\min_{c\in C_{i-1}} d(o_j,c)$.
\end{proof}

However, we remark that while we can guarantee the existence of  a good clustering using $C$ as centers, it seems hard to recover it while ensuring  fairness. This stands in contrast to unconstrained clustering, where one can simply assign every point to its closest center. For the special case $\ell=2$, a fair clustering may be recovered using flow-based techniques. For $\ell\geq 3$, deciding whether there exists a clustering with some cost, given a candidate set of centers, it is a hard problem. 
The proof is a simple reduction from the 3D matching problem. Since the reduction is similar to the Proposition~\ref{prop:hardness}, we omit details.

\begin{proposition}
Let $A$ be a set of $\ell\cdot n$ points in some finite metric with a fair coloring $c:[\ell\cdot n]\rightarrow [\ell]$, let $C$ be (a possibly optimal) set of $k$ centers and let $t>0$ be a parameter. Then deciding whether there exists a fair $k$-center clustering using $C$ as centers with the range $[t,3t]$ is NP-hard.
\end{proposition}

\section{A PTAS for Fair Clustering in Euclidean Spaces with Constant $k$}

Lastly, we briefly show how to derive a $(1+\varepsilon)$ approximation for fair $k$-clustering in Euclidean spaces if the number of centers is constant. 
This shows that a separation between the hardness of unconstrained clustering and fair clustering has to consider large values of $k$.

\begin{theorem}
Let $A$ be a set of points in Euclidean space and let $k$ be a constant. Then there exists an algorithm that computes in time $O(n^{\text{poly}(k/\varepsilon)})$ a $(1+\varepsilon)$ approximation for fair $k$-median, fair $k$-means, and fair $k$-center.
\end{theorem}
\begin{proof}
The high level idea is similar to early polynomial time approximation schemes for unconstrained $k$-clustering~\cite{AckermannBS08,BadoiuHI02}, with a few modifications to account for fairness. Assume we are given an oracle that (i) returns a set of $k$ centers such that these centers form a $(1+\varepsilon)$ approximation and (ii) returns the size of the clusters associated to these centers. 
If we have access to both, we can recover a clustering with the same approximation ratio by solving the following minimum transportation problem.
For every color, we construct an assignment as follows.
Every input point $p$ corresponds to a node $v_p$ in a flow network. Every center $c$ corresponds to a node $u_c$. These nodes are connected by a unit capacity edge. Furthermore, we have unit capacity edges from the source node to each $v_p$, as well as edges from the nodes $u_c$ to the target node. These edges have capacity that are exactly the target size of the clustering.
We now find a feasible flow such that the connection cost $\sum_v\sum_u f(v_p,u_c)\cdot d(p,c)$ is minimized, where $d(p,c)$ corresponds to the Euclidean distance between points $p$ and $c$~\footnote{For $k$-means, we would have to use squared Euclidean distances. For $k$-center, we would use a threshold network that only connects nodes to centers that are within distance $(1+\varepsilon)\cdot OPT$ and find an arbitrary flow.}. Finding a feasible flow can be done in polynomial time, moreover such a flow is integral, i.e. guaranteed to be a fair assignment.

To remove the oracle, we do the following. For (ii), we observe that there are $O(n^k)$ different ways of selecting the sizes of the $k$ clusters, given a ground set of $n$ points.
For (i), it is well known that for all of the considered objectives, there exist weak coresets of for a single center of size $\text{poly}(\varepsilon^{-1})$, see~\cite{BadoiuHI02} and~\cite{AckermannBS08}.
Weak coresets essentially satisfy the following property: Given a point set $A$, a weak coreset wrt to some objective is a subset of $S$ of $A$ such that a $(1+\varepsilon)$ approximation computed on $S$ is a $(1+O(\varepsilon))$ computed on $A$.

Hence, we can find a suitable set of $\text{poly}(k/\varepsilon^{-1})$ points from which to compute $k$ candidate centers by enumerating all $\text{poly}(k/\varepsilon^{-1})$-tuples in time $n^{\text{poly}(k/\varepsilon^{-1})}$.
\end{proof}

We complement this result by showing a fairlet decomposition is APX-hard for $\ell\geq 3$. In particular, we also show that computing a better than $2$-approximate $n$-center clustering decomposition is NP-hard for $\ell\geq 3$. 
%Here we use that $3$-dimensional matching is APX-hard, see for instance Chleb{\'{\i}}k and Chleb{\'{\i}}kov{\'{a}}~\cite{ChlebikC06}. 
Hence, the analysis of Theorem~\ref{thm:approxclustering} is tight.
If a better approximation algorithm for fair clustering exists, it will have to rely on a different technique.
Note that this stands in contrast to the computability of an optimal fairlet decomposition for $\ell=2$ colors proposed by~\cite{CKLV17}.

\begin{proposition}
\label{prop:hardness}
Let $A$ be a set of $\ell\times d$ points in a finite metric, $\ell\geq 3$, let $c:[\ell\cdot n]\rightarrow [\ell]$ be a balanced coloring of $A$. Then approximating fair $n$-center beyond a factor of $2$
and approximation fair $n$-median beyond a factor of $\frac{96}{95}$ is $NP$-hard. 
\end{proposition}
\begin{proof}
We give a reduction from $3$-dimensional matching to fair $k$-center with three colors (a generalization from $\ell$-dimensional matching and $\ell$ colors is straightforward). Given a hypergraph $G(X\uplus Y \uplus Z,E)$, with disjoint nodes sets $X,Y,Z$ of size $n$ each and $k$ hyperedges $E\subset X\times Y\times Z$, $3$-dimensional matching consists of deciding whether there exists perfect hypermatching, i.e. a collection of $n$ pairwise disjoint hyperedges $H\subseteq E$. 

We construct an instance of fair $k$-center as follows. Each hyperedge $e\in E$ will be mapped to some point $p_e$ and also every node $v\in X,Y,Z$ will be mapped to some point $p_v$. The points corresponding to hyperedges will be our candidate set of centers $C$. We now define the distances between our points as follows. For nodes $v$ and hyperedges $e$, we set $d(p_v,p_e)=\begin{cases}1 &\text{if }v\in e \\
2 &\text{if } v\notin e \end{cases}$. 
The remaining distances are set to $2$.
This trivially results in a metric.

%To see that this induces a metric, we only need to observe that the triangle inequality holds for $d(p_v,p_e)$ when $v\notin e$. For any $u\in V\setminus \{v\}$, we have $d(p_v,p_u)+d(p_u,p_e)= 2 + d(p_u,p_e) \geq 3 = d(p_v,p_e)$. For any $e'\in E\setminus \{e\}$, we have $d(p_v,p_{e'}) + d(p_{e'},p_{e}) = d(p_v,p_{e'}) + 2 \geq 3 = d(p_v,p_e)$.

Now, assume that a perfect hypermatching exists. Then the fair $n$-center clustering cost is precisely $1$. If, however, no perfect hypermatching exists, the cost is $2$. Distinguishing between these two cases is NP-hard, hence approximating fair $n$-center beyond a factor is also $NP$ hard.

Similary, if a perfect hypermatching exists, the cost of a fair $n$-median clustering is precisely $3n$. If the size of the largest hypermatching is $3n-t$, then at least $t$ points have to pay $2$, i.e. the total cost is at least $3n + t$. Since distinguishing between a perfect hypermatching and a hypermatching of size $95/94$~\cite{ChlebikC03}, this implies $t\geq \frac{1}{95}n$ and therefore APX hardness beyond a factor $96/95$
%\paragraph{Soundness} Let $H$ be a perfect hypermatching. We show that this induces a fair $k$-center clustering of cost $1$. For every hyperedge $e=(x,y,z)\in H$, we set the cluster $C_e:=\{p_x,p_y,p_z\}$ with center $p_e$, the remaining clusters are empty. Clearly, the resulting clustering is fair and the distance of every point to its assigned center is $1$. 
%
%\paragraph{Completeness} Suppose no perfect hypermatching exists. We show that then there exists no fair clustering of cost $1$ (or even less than $2-\varepsilon$) using the $p_e$ as centers. For the sake of contradiction, suppose there exists such a clustering with clusters $C_1,\ldots C_k$. Any for any cluster with center $p_e$ and more than three nodes must contain a point $p_v$ such that $v\notin e$, and hence have cost $3$. Hence all clusters must have size at most $3$. The union of these clusters, however, is a perfect hypermatching.
\end{proof}

\section{Experimental Analysis}
\label{se:experiments}

\subsection{Dataset Description}

The datasets used for experiments are taken from the previous literature \cite{BCN19, CKLV17, BIOSVW19}. 
As our interest is in the multiple-color scenario, we ran our experiments considering 8 colors; for completeness, we also consider similar experiments with 4 colors in the supplementary material. Each color represents a protected class, characterized by some particular value of the chosen protected attributes.
We selected protected attributes to obtain 8 classes in total, and we also subsampled the original records for obtaining the same number of records for each class.
In total, we used the six data sets.
We report averages computed over 100 samples of 1000 distinct points. 
Each sample is a perfectly balanced set of points with respect to the eight colors described in the following.

\paragraph{Adults}
This dataset\footnote{https://archive.ics.uci.edu/ml/datasets/Adult} 
contains ``1994 US census'' records about registered individuals including 
age, education, marital status, occupation, ethnicity, sex, hours worked per week, 
native country, and others. Following \cite{BCN19} and \cite{CKLV17}, the 
numerical attributes chosen to represent points in the Euclidean space are 
\texttt{age, fnlwgt, education-num, capital-gain, hours-per-week}. The protected 
attributes chosen to represent the classes are \texttt{Sex}, \texttt{Ethnicity}, \texttt{Income}, where each of them takes only 2 possible values.
For the experiments, we used 100 balanced subsamples of 1000 distinct records.

\paragraph{Athletes}
This dataset\footnote{www.kaggle.com/heesoo37/120-years-of-olympic-history-athletes-and-results} 
contains bio data on Olympic athletes and medal results from Athens 1896 to Rio 2016. 
The selected features are \texttt{Age}, \texttt{Height}, \texttt{Weight}. 
The protected attributes are \texttt{Sex}, \texttt{Sport}, \texttt{Medal} (two sports 
were selected - gymnastics and basketball - and two types of athletes were considered 
for the third attribute - athletes who won at least one medal and athletes who did not). 
For the experiments, we used 100 balanced subsamples of 1000 distinct records.
%1000 records are sampled with the same rule of obtaining a balanced dataset w.r.t. the 
%protected attributes.

\paragraph{Bank}
This dataset\footnote{https://archive.ics.uci.edu/ml/datasets/Bank+Marketing} 
stems from direct marketing campaigns, based on phone calls, of a Portuguese banking institution. 
As in \cite{BCN19} and \cite{CKLV17}, the selected features to 
represent the points in the space are \texttt{age, balance, duration}. The protected 
attributes are \texttt{marital status} (married or not), \texttt{education} (secondary 
or tertiary), \texttt{housing}. 
%As usual, in the 4-colors case the third was dropped. 
%Also here, 100 subsamples of 1000 records were used for the experiments.
For the experiments, we used 100 balanced subsamples of 1000 distinct records.

\paragraph{Diabetes}
The dataset\footnote{https://archive.ics.uci.edu/ml/datasets/Diabetes+130-US+hospitals+for+years+1999-2008}, 
used for experiments in \cite{CKLV17}, represents 10 years (1999-2008) of 
clinical care at 130 US hospitals and integrated delivery networks. 
It includes over 50 
features representing patient and hospital outcomes; of these features, 4 were chosen to 
represent the points in the space: \texttt{time\_in\_hospital, num\_lab\_procedures, 
num\_medications, number\_diagnoses}.
The protected attributes are \texttt{sex}, \texttt{ethnicity} (\texttt{Caucasian} or 
\texttt{AfricanAmerican}), \texttt{age} (this attribute has been dichotomised in order to have 
two classes of ages: people who are respectively less and more than 50 years old). 
We ran the experiments on 100 subsamples of 1000 distinct records each.
For the experiments, we used 100 balanced subsamples of 1000 distinct records.

\paragraph{Credit cards}
This dataset\footnote{https://archive.ics.uci.edu/ml/datasets/default+of+credit+card+clients}, 
contains information on credit card holders from a certain credit card in Taiwan. 
Here, the same 14 features chosen by \cite{BCN19} were selected, while the 
protected attributes are \texttt{sex}, \texttt{education} (\texttt{Graduate School} or 
\texttt{University}), \texttt{marriage} (married or not).
%, the last of which was dropped for the 4-colors experiments. 
%As usual, 100 subsamples of 1000 distinct records were used.
For the experiments, we used 100 balanced subsamples of 1000 distinct records.

\paragraph{CensusII}
This dataset contains records extracted from the \textit{USCensus1990raw}\footnote{https://archive.ics.uci.edu/ml/datasets/US+Census+Data+\%281990\%29} 
data set (also used in~\cite{BIOSVW19}), containing 2458285 records composed by 68 attributes.
Among all of these attributes, 9 have been chosen to represent the points in the Euclidean space: \texttt{AGE, AVAIL, CITIZEN, CLASS, DEPART, HOUR89, HOURS, PWGT1, TRAVTIME}.
For this dataset, the selected protected attributes are \texttt{SEX} (\texttt{Female}, \texttt{Male}), \texttt{RACE} (dichotomized as \texttt{White}, \texttt{notWhite}) and \texttt{MARITAL} (dichotomized as \texttt{NowMarried}, \texttt{NowNotMarried}).
For the experiments, we used 100 balanced subsamples of 1000 and 450000 distinct records.

\subsection{Setup and Algorithms}

We solved the fair $k$-median problem by implementing Algorithms~\ref{alg:reduction},~\ref{alg:randmedian}, 
\textbf{Q} and \textbf{Excellent}.
\textbf{Q} is similar to Algorithm~\ref{alg:reduction}, except that we select the color with minimum perfect matching cost. This algorithm is guaranteed to return an $(\alpha+4)$ approximation and does so slightly faster than Algorithm~\ref{alg:reduction}. 
\textbf{Excellent} is a further variant of Algorithm~\ref{alg:reduction} that computes a good clustering for each color and subsequently performs a fair assignment. The approximation factor is theoretically equal to that of Algorithm~\ref{alg:reduction}, but sometimes improves empirically.

We ran the algorithms for all values of $k$ between $2$ and $20$. For $1$ center, any solution is naturally fair. Since there was already little to no difference between the cost of a fair $20$ clustering and the cost of a fair $n$ clustering, we did not consider larger values of $k$.

We compared these algorithms with the implementation of~\cite{BCN19}.
For the largest data set (\textit{USCensus1990raw}) consisting of $8$ colors with a total of $450000$ points, the code by~\cite{BCN19} did not terminate. 
On this dataset, we showcased the modularity of our approach by combining it with the fast fairlet algorithm by~\cite{BIOSVW19}.

Since our algorithm requires a solver for the unconstrained $k$-median problem, 
For all 1000-points datasets, we used the single-swap local search heuristic, while yields a 5-approximation in the worst case~\cite{AryaGKMMP04}.

For the 450000-points \textit{USCensus1990raw} dataset, local search is infeasible to run. Instead, we used a simple heuristic that essentially mimics the $k$-means++ algorithm~\cite{AV07}: First we sample $k$ centers by iteratively picking the next center proportionate to its distance to the previously chosen centers, and then running the $k$-medoids algorithm to further refine the solution.
For the experiments we used a Intel Xeon 2.4Ghz with 24GB of RAM and a Linux Ubuntu 18.04 LTS Operating System.

\subsection{Results}

The left plot of the following figures reports the aggregated average cost of all tested methods; the right plot reports running times. 
In addition the fair clustering algorithms, we also reported costs for fair $n$ clustering, which provides a lower bound for any fairlet-based algorithm, as well as the cost of an unconstrained solution.

For the most part, the algorithm by~\cite{BCN19} has comparable cost to our algorithms. Furthermore, we empirically observe that~\cite{BCN19} almost always computes a balanced solution, as opposed to the bicriteria result sated in their paper. 
Specifically, less than 0.8\% of instances for \textit{Diabetes}, less than 0.6\% of instances for \textit{Credit cards} dataset, less than 0.3\% of instances for \textit{USCensus1990raw} dataset, less than 0.2\% of instances for \textit{Athletes} dataset, and less than 0.05\% of instances for \textit{Bank} yielding an unfair solution.
Our algorithms, of course, always guarantee fairness. 
All of our algorithms perform slightly better than~\cite{BCN19} on data sets in which the fairlets (i.e. the fair $n$ clusterings) are very cheap compared to the cost of a $k$-clustering, see Figure~\ref{fig:exp-1-k-median-diabetes-8-colors}. 
On data sets, where the fairlets are more expensive their is little difference in cost, see Figure~\ref{fig:exp-1-k-median-censusII-8-colors}.

In terms of running time, all of our algorithms run substantially faster than~\cite{BCN19} by roughly factors of 100 or more. Algorithms~\ref{alg:randmedian} and {\bf Q} have an average running time of 176msc and 258msec, respectively. This is also significantly faster than Algorithms~\ref{alg:reduction} and {\bf Excellent} (1111msec and 1044msc on average respectively), while having a roughly comparable cost.

\begin{figure*}[h!]
\subfloat[]{\includegraphics[width=.6125\textwidth]{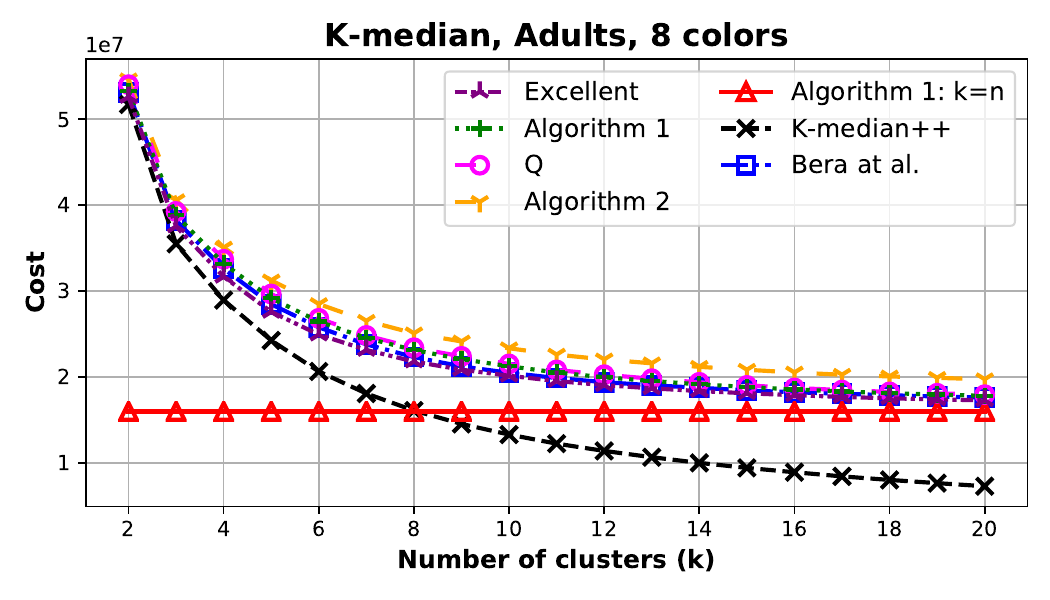}}
\subfloat[]{\includegraphics[width=.35\textwidth]{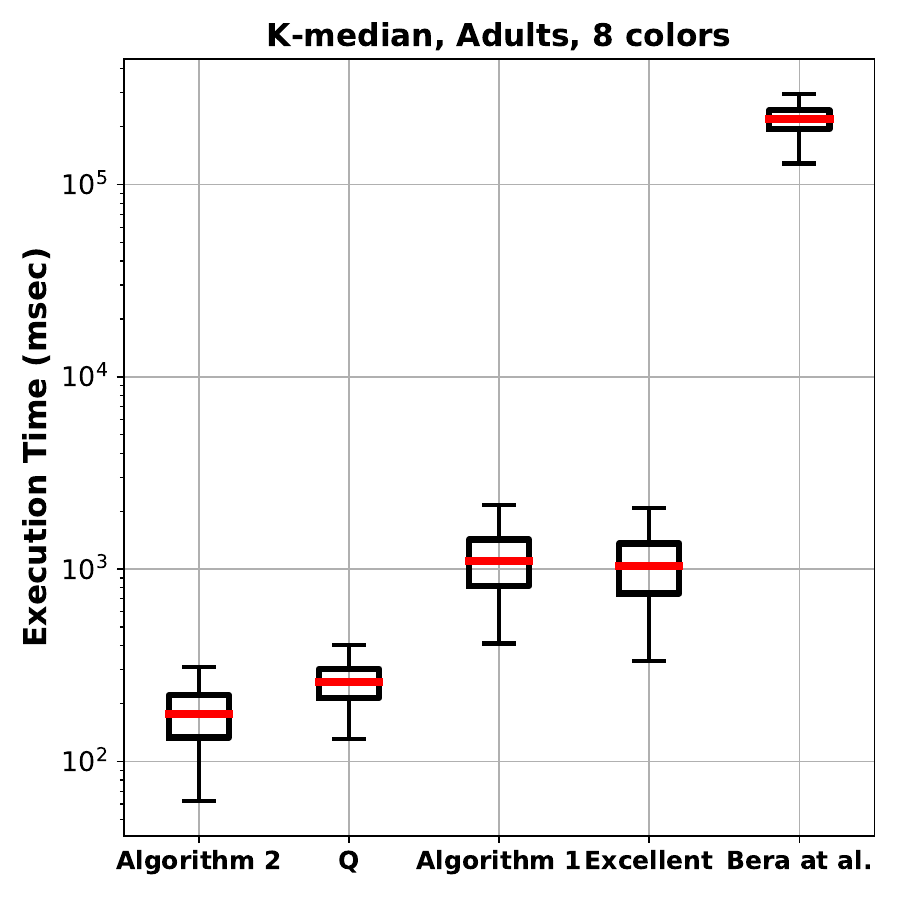}}
\caption{Average cost and execution time of the fair-k-median methods on \textit{Adults} dataset: 8 colors,	 100 subsamples of 1000 distinct points each.}
\label{fig:exp-1-k-median-adults-8-colors}
\end{figure*}
\begin{figure*}[h!]
\subfloat[]{\includegraphics[width=.6125\textwidth]{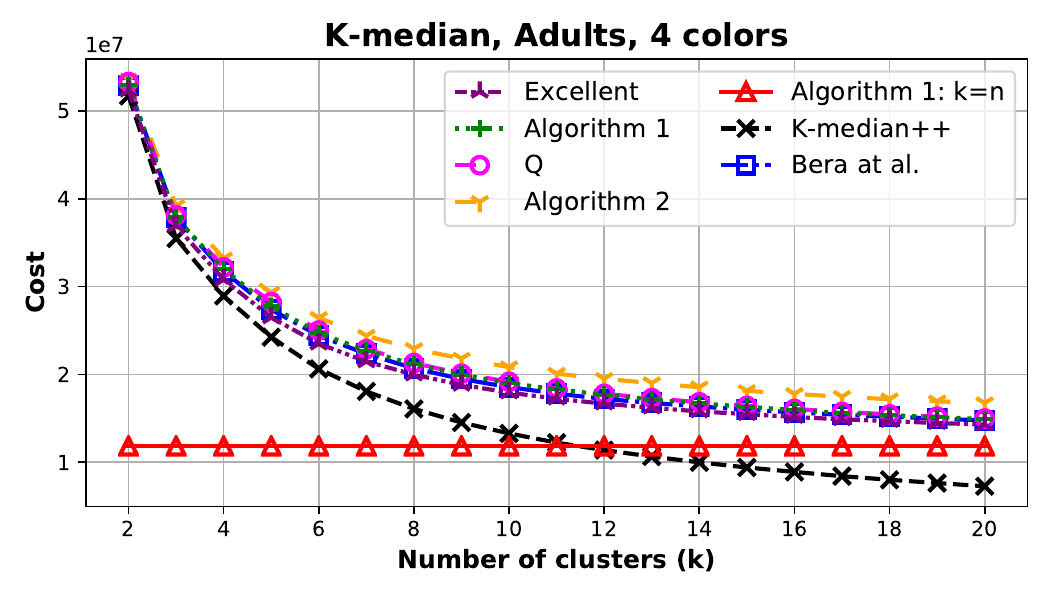}}
\subfloat[]{\includegraphics[width=.35\textwidth]{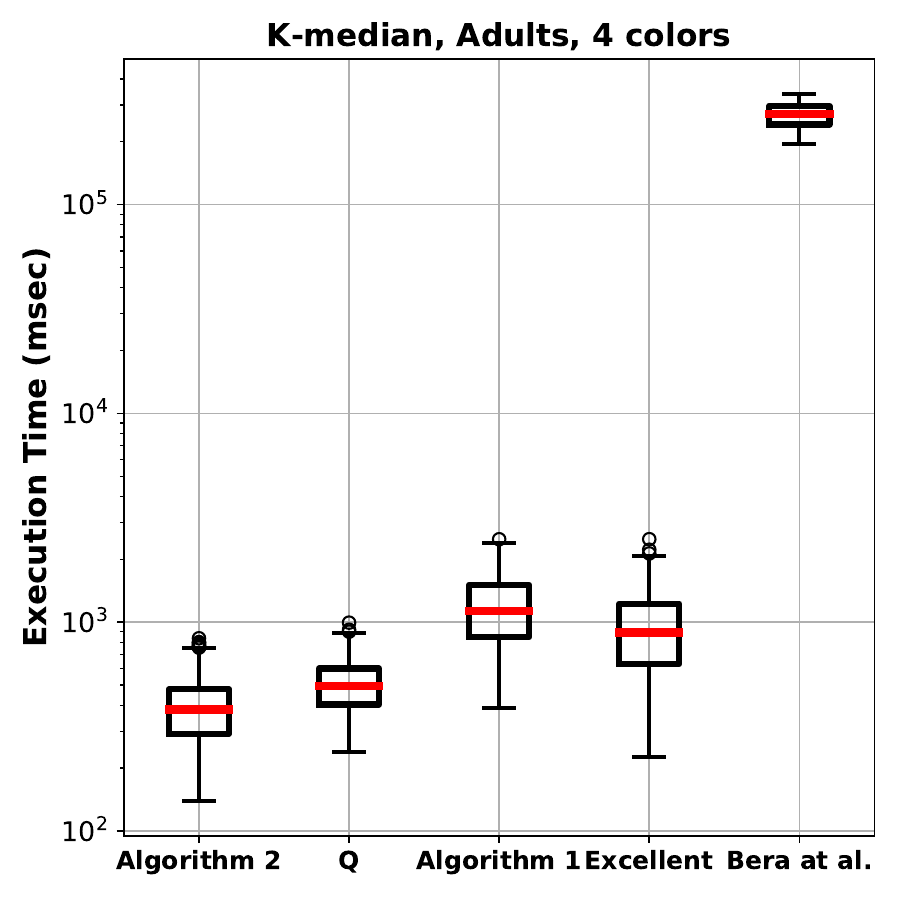}}
\caption{Average cost and execution time of the fair-k-median methods on \textit{Adults} dataset: 4 colors,	 100 subsamples of 1000 distinct points each.}
\label{fig:exp-1-k-median-adults-4-colors}
\end{figure*}

\begin{table*}[h!]
\begin{center}
\begin{tabular}{|c|c|c|c|}
\hline

            & $k \in [2, 5]$ & $k \in [6, 10]$ & $k \in [11, 20]$ \\ \hline
 {\bf Algorithm~\ref{alg:reduction}: k=n}  & 15981152.5 \textit{1968761.8} & 15981152.5 \textit{1968761.8} & 15981152.5 \textit{1968761.8} \\ \hline
 {\bf k-median++}  & 35093729.4 \textit{10482181.7} & 16518716.6 \textit{2704863.2} & 9390692.2 \textit{1632344.3} \\ \hline
 {\bf Excellent}  & 37334601.9 \textit{9646426.1} & 22164330.5 \textit{2448591.7} & 18136890.1 \textit{2030020.3} \\ \hline
 {\bf Algorithm~\ref{alg:reduction}}  & 38601051.2 \textit{9319399.3} & 23480401.1 \textit{2606059.3} & 18863910.7 \textit{2110809.0} \\ \hline
 {\bf Q}  & 39150199.6 \textit{9445281.9} & 23786205.0 \textit{2731965.7} & 19069140.1 \textit{2222137.3} \\ \hline
 {\bf Algorithm~\ref{alg:randmedian}}  & 40372180.0 \textit{9277749.3} & 25507540.8 \textit{3548581.9} & 20864471.0 \textit{3425165.9} \\ \hline
 {\bf ~\cite{BCN19}}  & 38101545.1 \textit{9528290.7} & 22713743.0 \textit{2587922.8} & 18488295.4 \textit{2018105.9} \\ \hline

\end{tabular}
\caption{Average and standard deviation of the cost of the fair-k-median methods on \textit{Adults} dataset: 8 colors, 100 subsamples of 1000 distinct points each.}
\label{tab:adults-error-8}
\end{center}
\end{table*}

\begin{table*}[h!]
\begin{center}
\begin{tabular}{|c|c|c|c|}
\hline

            & $k \in [2, 5]$ & $k \in [6, 10]$ & $k \in [11, 20]$ \\ \hline
 {\bf Algorithm~\ref{alg:reduction}: k=n}  & 11839345.7 \textit{2070002.0} & 11839345.7 \textit{2070002.0} & 11839345.7 \textit{2070002.0} \\ \hline
 {\bf k-median++}  & 35093729.4 \textit{10482181.7} & 16518716.6 \textit{2704863.2} & 9390692.2 \textit{1632344.3} \\ \hline
 {\bf Excellent}  & 36625310.1 \textit{9921526.0} & 20339741.6 \textit{2570564.3} & 15488437.9 \textit{2061712.3} \\ \hline
 {\bf Algorithm~\ref{alg:reduction}}  & 37621394.0 \textit{9712253.2} & 21512025.2 \textit{2700895.7} & 16314405.3 \textit{2192492.3} \\ \hline
 {\bf Q}  & 37958840.6 \textit{9713117.1} & 21687331.5 \textit{2756245.2} & 16417451.2 \textit{2213958.7} \\ \hline
 {\bf Algorithm~\ref{alg:randmedian}}  & 38834788.7 \textit{9553619.9} & 23278062.9 \textit{3622493.8} & 18138272.2 \textit{3406247.1} \\ \hline
 {\bf ~\cite{BCN19}}  & 37450212.1 \textit{9797625.9} & 21096203.8 \textit{2680446.0} & 16013188.3 \textit{2054989.7} \\ \hline

\end{tabular}
\caption{Average and standard deviation of the cost of the fair-k-median methods on \textit{Adults} dataset: 4 colors, 100 subsamples of 1000 distinct points each.}
\label{tab:adults-error-4}
\end{center}
\end{table*}

%%%%%%%%%%%%%%%%
%%% ATHLETES %%%
\begin{figure*}[h!]
\subfloat[]{\includegraphics[width=.6125\textwidth]{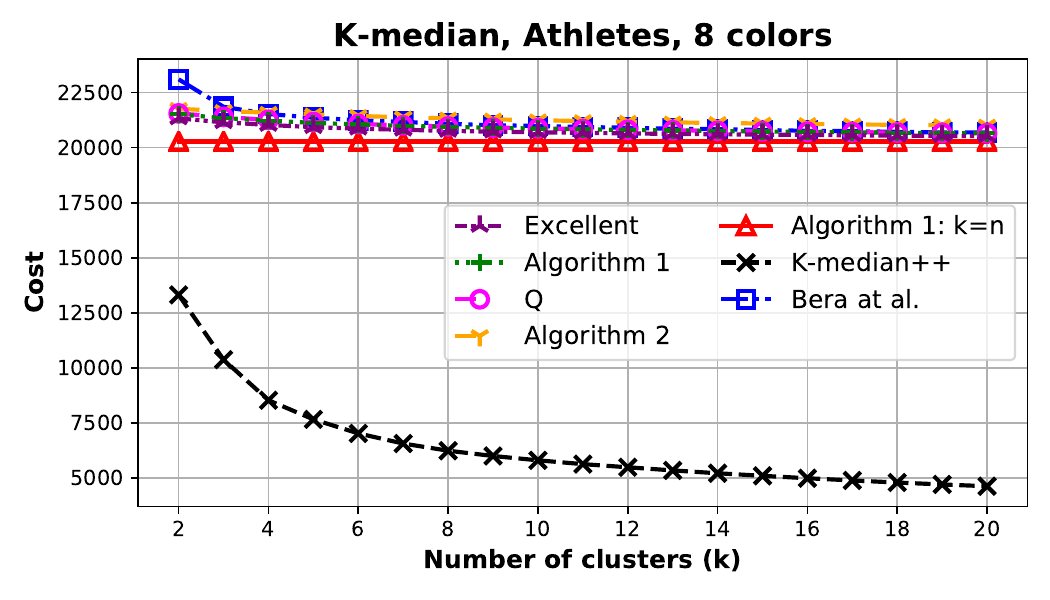}}
\subfloat[]{\includegraphics[width=.35\textwidth]{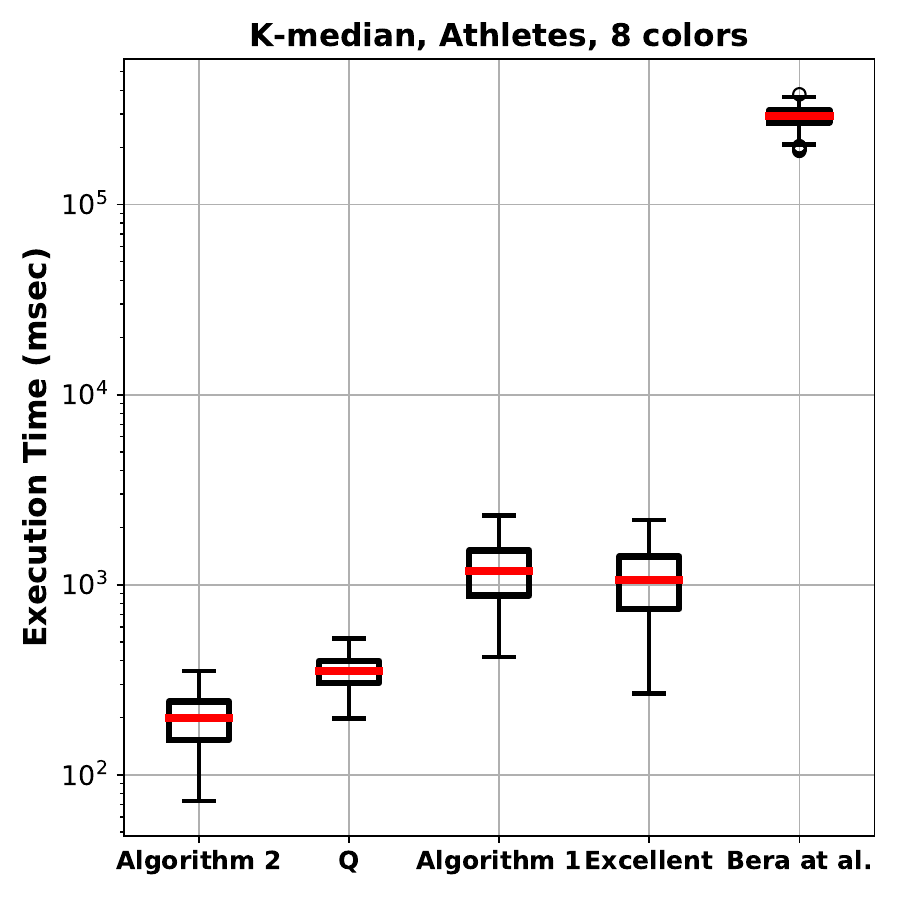}}
\caption{Average cost and execution time of the fair-k-median methods on \textit{Athletes} dataset: 8 colors, 100 subsamples of 1000 distinct points each.}
\label{fig:exp-1-k-median-athletes-8-colors}
\end{figure*}
\begin{figure*}[h!]
\subfloat[]{\includegraphics[width=.6125\textwidth]{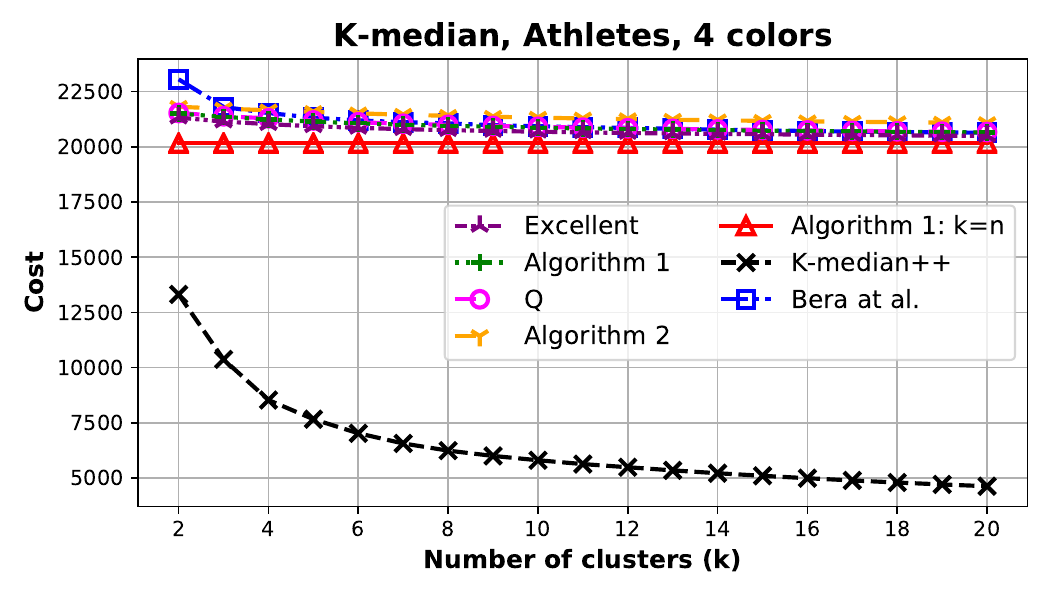}}
\subfloat[]{\includegraphics[width=.35\textwidth]{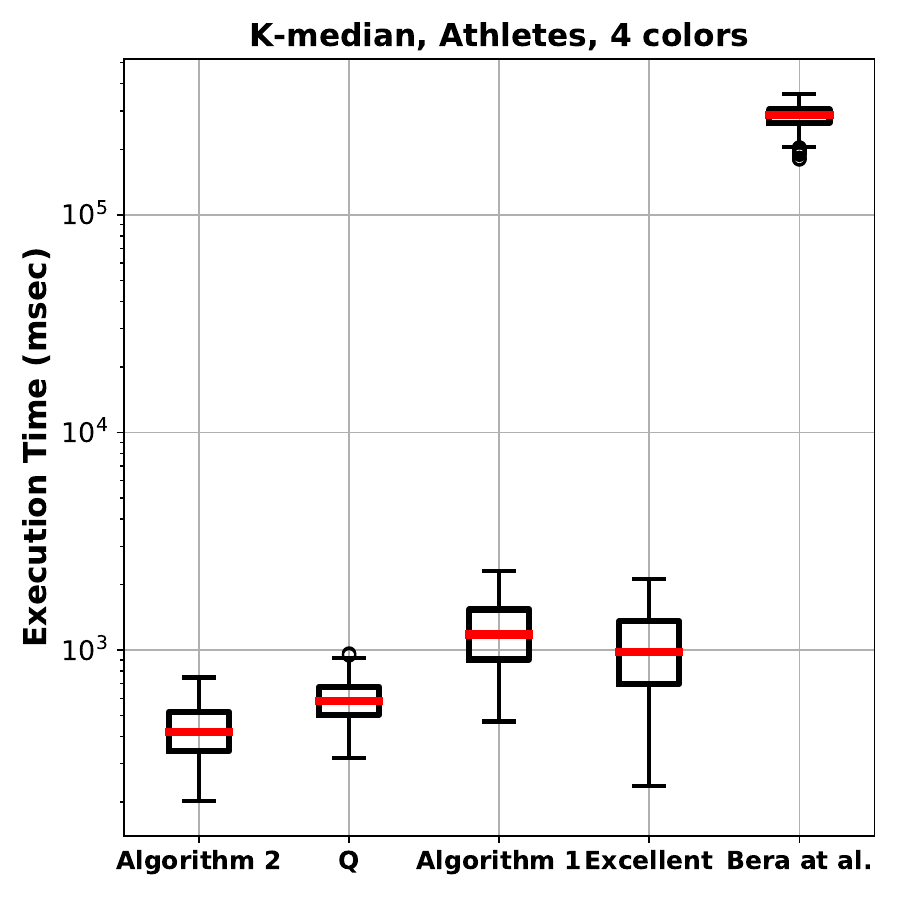}}
\caption{Average cost and execution time of the fair-k-median methods on \textit{Athletes} dataset: 4 colors, 100 subsamples of 1000 distinct points each.}
\label{fig:exp-1-k-median-athletes-4-colors}
\end{figure*}

\begin{table*}[h!]
\begin{center}
\begin{tabular}{|c|c|c|c|}
\hline

            & $k \in [2, 5]$ & $k \in [6, 10]$ & $k \in [11, 20]$ \\ \hline
 {\bf Algorithm~\ref{alg:reduction}: k=n}  & 20278.4 \textit{280.9} & 20278.4 \textit{280.9} & 20278.4 \textit{280.9} \\ \hline
 {\bf k-median++}  & 9970.3 \textit{2171.1} & 6327.1 \textit{439.9} & 5079.1 \textit{327.3} \\ \hline
 {\bf Excellent}  & 21107.9 \textit{300.0} & 20769.6 \textit{276.7} & 20585.8 \textit{273.5} \\ \hline
 {\bf Algorithm~\ref{alg:reduction}}  & 21299.8 \textit{297.6} & 20948.8 \textit{276.4} & 20728.2 \textit{275.6} \\ \hline
 {\bf Q}  & 21340.9 \textit{302.0} & 20965.2 \textit{278.8} & 20736.5 \textit{276.4} \\ \hline
 {\bf Algorithm~\ref{alg:randmedian}}  & 21629.2 \textit{391.5} & 21333.2 \textit{418.7} & 21092.3 \textit{411.6} \\ \hline
 {\bf ~\cite{BCN19}}  & 21959.2 \textit{796.3} & 21098.5 \textit{295.8} & 20793.4 \textit{288.7} \\ \hline

\end{tabular}
\caption{Average and standard deviation of the cost of the fair-k-median methods on \textit{Athletes} dataset: 8 colors, 100 subsamples of 1000 distinct points each.}
\label{tab:athletes-error-8}
\end{center}
\end{table*}

\begin{table*}[h!]
\begin{center}
\begin{tabular}{|c|c|c|c|}
\hline

            & $k \in [2, 5]$ & $k \in [6, 10]$ & $k \in [11, 20]$ \\ \hline
 {\bf Algorithm~\ref{alg:reduction}: k=n}  & 20157.3 \textit{279.6} & 20157.3 \textit{279.6} & 20157.3 \textit{279.6} \\ \hline
 {\bf k-median++}  & 9970.3 \textit{2171.1} & 6327.1 \textit{439.9} & 5079.1 \textit{327.3} \\ \hline
 {\bf Excellent}  & 21105.5 \textit{302.2} & 20763.7 \textit{275.0} & 20564.7 \textit{275.7} \\ \hline
 {\bf Algorithm~\ref{alg:reduction}}  & 21311.2 \textit{298.7} & 20959.8 \textit{278.5} & 20731.3 \textit{277.7} \\ \hline
 {\bf Q}  & 21349.7 \textit{299.4} & 20999.8 \textit{289.8} & 20758.7 \textit{285.3} \\ \hline
 {\bf Algorithm~\ref{alg:randmedian}}  & 21694.3 \textit{384.7} & 21410.3 \textit{416.6} & 21171.7 \textit{405.5} \\ \hline
 {\bf ~\cite{BCN19}}  & 21927.4 \textit{790.9} & 21058.0 \textit{300.0} & 20744.4 \textit{288.3} \\ \hline

\end{tabular}
\caption{Average and standard deviation of the cost of the fair-k-median methods on \textit{Athletes} dataset: 4 colors, 100 subsamples of 1000 distinct points each.}
\label{tab:athletes-error-4}
\end{center}
\end{table*}

%%%%%%%%%%%%
%%% BANK %%%
\begin{figure*}[h!]
\subfloat[]{\includegraphics[width=.6125\textwidth]{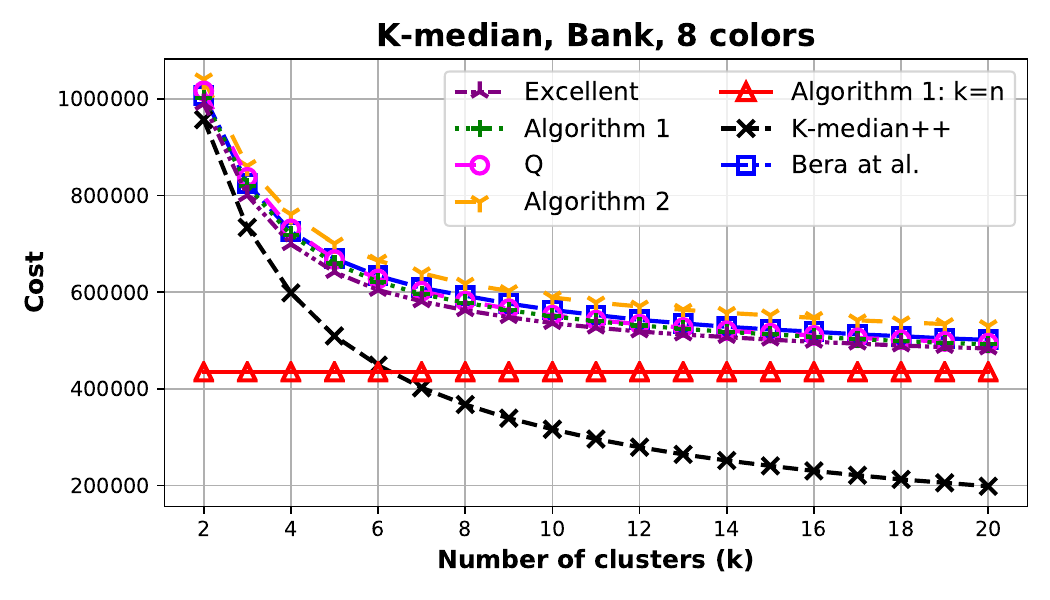}}
\subfloat[]{\includegraphics[width=.35\textwidth]{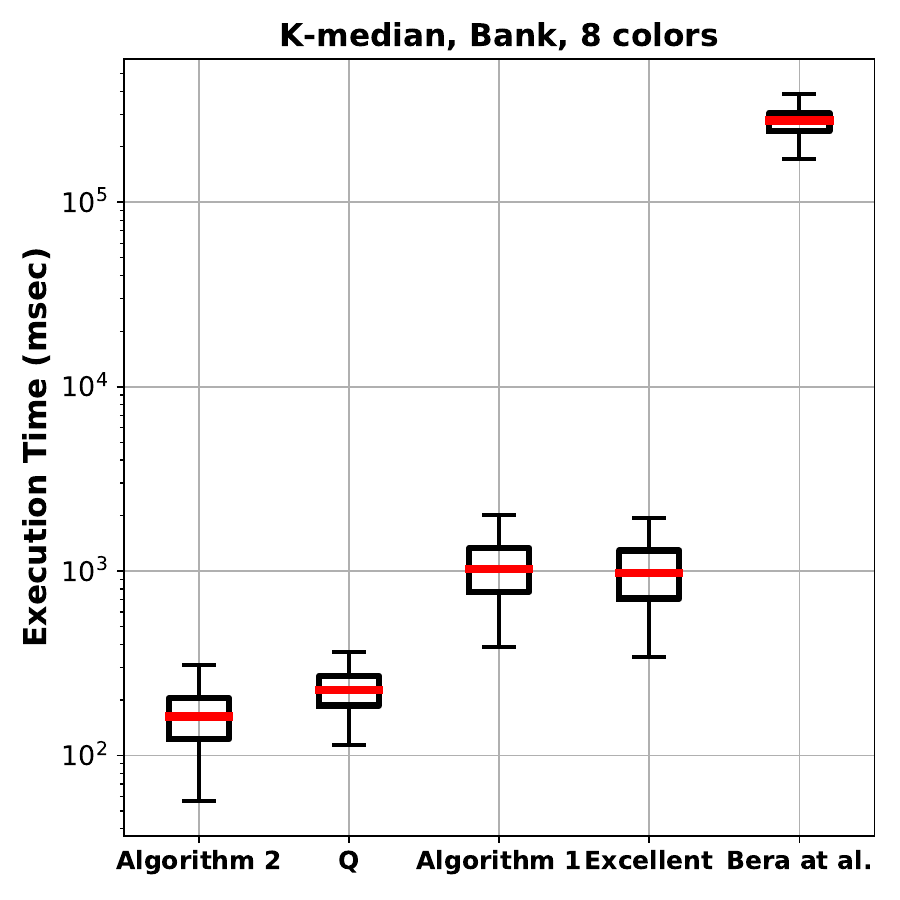}}
\caption{Average cost and execution time of the fair-k-median methods on \textit{Bank} dataset: 8 colors, 100 subsamples of 1000 distinct points each.}
\label{fig:exp-1-k-median-bank-8-colors}
\end{figure*}
\begin{figure*}[h!]
\subfloat[]{\includegraphics[width=.6125\textwidth]{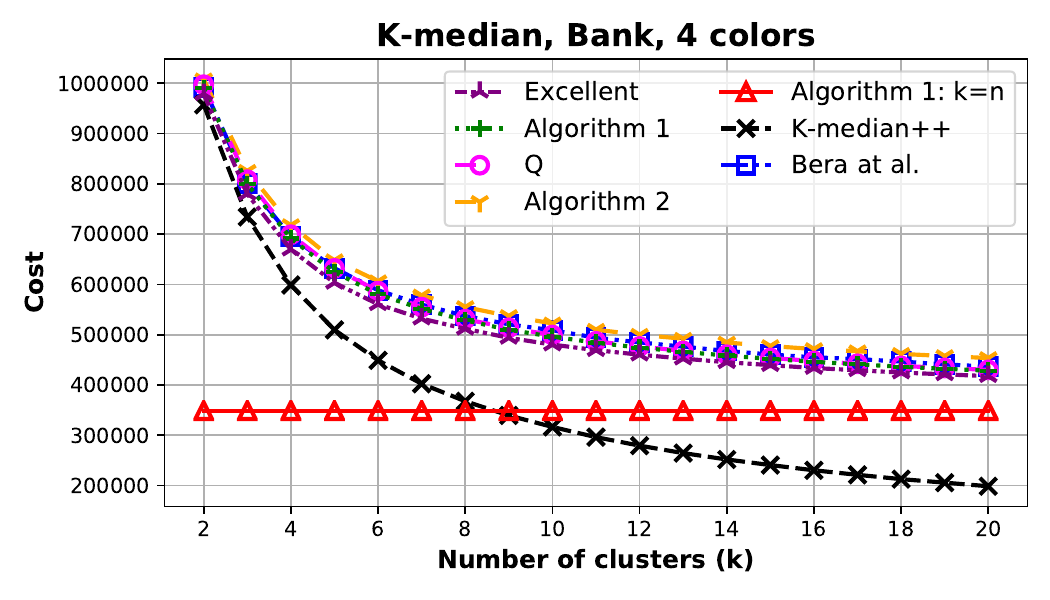}}
\subfloat[]{\includegraphics[width=.35\textwidth]{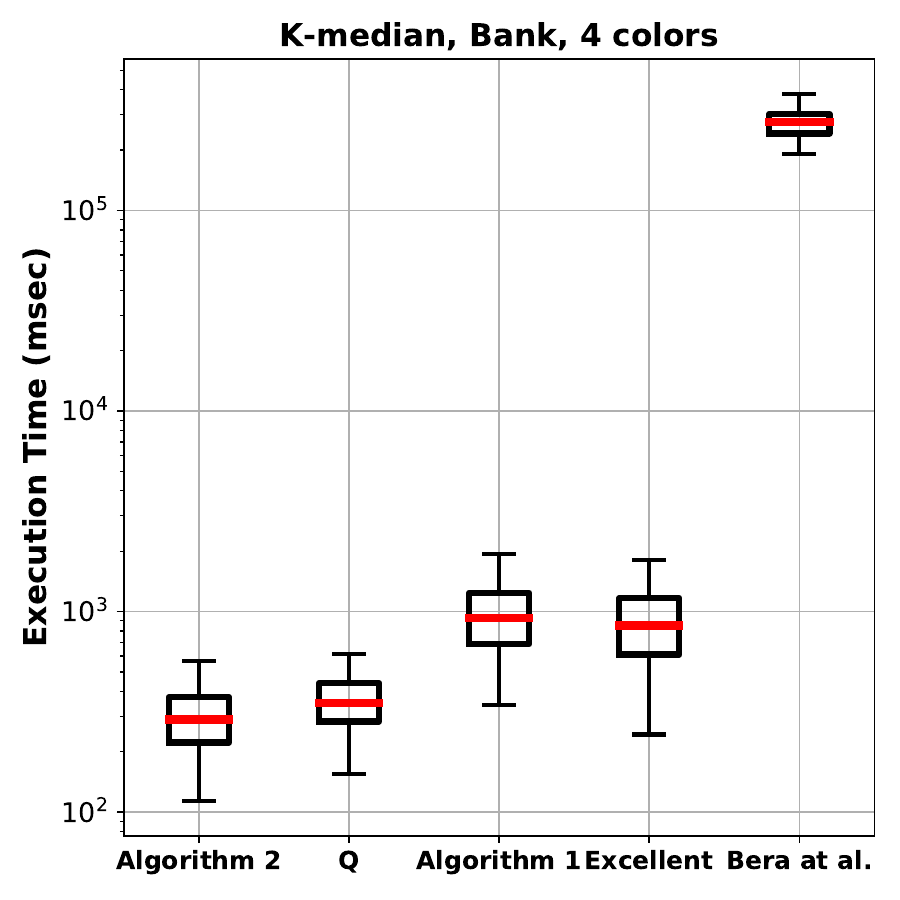}}
\caption{Average cost and execution time of the fair-k-median methods on \textit{Bank} dataset: 4 colors, 100 subsamples of 1000 distinct points each.}
\label{fig:exp-1-k-median-bank-4-colors}
\end{figure*}

\begin{table*}[h!]
\begin{center}
\begin{tabular}{|c|c|c|c|}
\hline

            & $k \in [2, 5]$ & $k \in [6, 10]$ & $k \in [11, 20]$ \\ \hline
 {\bf Algorithm~\ref{alg:reduction}: k=n}  & 434892.4 \textit{63810.1} & 434892.4 \textit{63810.1} & 434892.4 \textit{63810.1} \\ \hline
 {\bf k-median++}  & 699773.5 \textit{178243.2} & 375005.8 \textit{53842.9} & 240197.8 \textit{33496.6} \\ \hline
 {\bf Excellent}  & 782666.7 \textit{148502.7} & 566505.9 \textit{66761.8} & 501728.7 \textit{63811.2} \\ \hline
 {\bf Algorithm~\ref{alg:reduction}}  & 799852.1 \textit{147306.8} & 581807.0 \textit{68002.9} & 512378.1 \textit{64506.9} \\ \hline
 {\bf Q}  & 813091.6 \textit{151845.3} & 585861.4 \textit{68185.5} & 514900.4 \textit{64694.5} \\ \hline
 {\bf Algorithm~\ref{alg:randmedian}}  & 840799.3 \textit{158718.3} & 622778.7 \textit{89740.2} & 551295.3 \textit{86282.3} \\ \hline
 {\bf ~\cite{BCN19}}  & 806863.1 \textit{146858.0} & 595020.8 \textit{70544.2} & 522864.6 \textit{65312.6} \\ \hline

\end{tabular}
\caption{Average and standard deviation of the cost of the fair-k-median methods on \textit{Bank} dataset: 8 colors, 100 subsamples of 1000 distinct points each.}
\label{tab:bank-error-8}
\end{center}
\end{table*}

\begin{table*}[h!]
\begin{center}
\begin{tabular}{|c|c|c|c|}
\hline

            & $k \in [2, 5]$ & $k \in [6, 10]$ & $k \in [11, 20]$ \\ \hline
 {\bf Algorithm~\ref{alg:reduction}: k=n}  & 348806.2 \textit{68568.3} & 348806.2 \textit{68568.3} & 348806.2 \textit{68568.3} \\ \hline
 {\bf k-median++}  & 699773.5 \textit{178243.2} & 375005.8 \textit{53842.9} & 240197.8 \textit{33496.6} \\ \hline
 {\bf Excellent}  & 757508.6 \textit{157020.1} & 515590.9 \textit{67641.0} & 439276.2 \textit{65541.3} \\ \hline
 {\bf Algorithm~\ref{alg:reduction}}  & 776267.5 \textit{154638.7} & 533173.4 \textit{69012.5} & 451817.3 \textit{65885.7} \\ \hline
 {\bf Q}  & 783853.3 \textit{154205.1} & 536940.7 \textit{68463.6} & 453614.1 \textit{65842.6} \\ \hline
 {\bf Algorithm~\ref{alg:randmedian}}  & 799761.9 \textit{154970.7} & 559700.2 \textit{73057.8} & 477443.6 \textit{72098.2} \\ \hline
 {\bf ~\cite{BCN19}}  & 779977.8 \textit{153222.8} & 542883.4 \textit{72216.1} & 461341.9 \textit{66716.7} \\ \hline

\end{tabular}
\caption{Average and standard deviation of the cost of the fair-k-median methods on \textit{Bank} dataset: 4 colors, 100 subsamples of 1000 distinct points each.}
\label{tab:bank-error-4}
\end{center}
\end{table*}

%%%%%%%%%%%%%%%%%%%
%%% CensusII 1K %%%
\begin{figure*}[h!]
\subfloat[]{\includegraphics[width=.6125\textwidth]{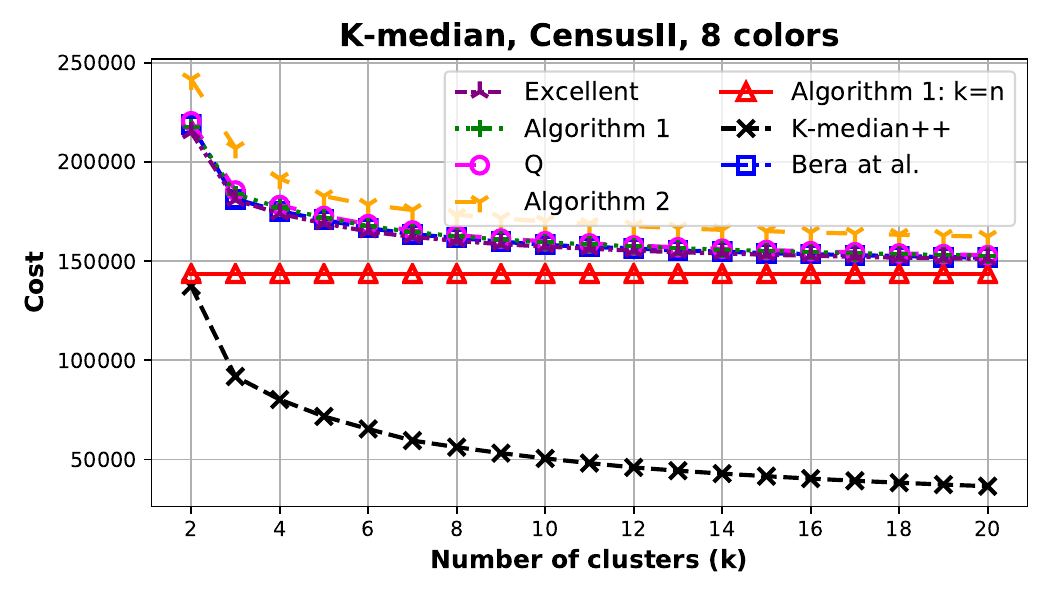}}
\subfloat[]{\includegraphics[width=.35\textwidth]{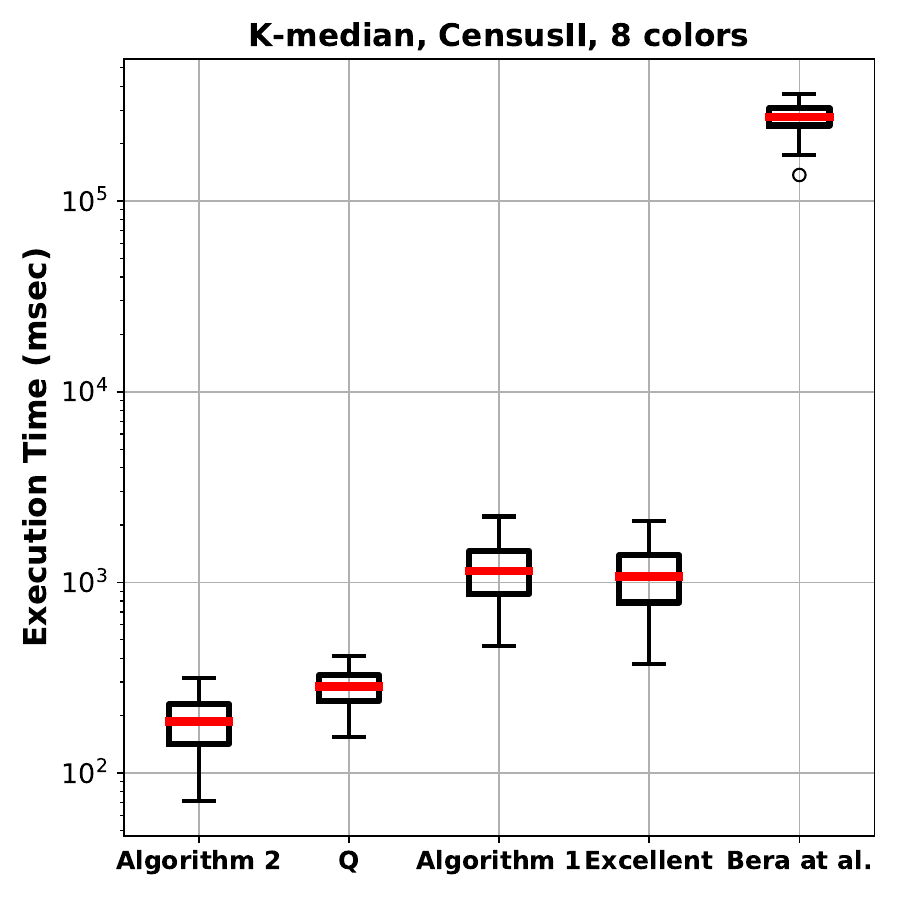}}
\caption{Average cost and execution time of the fair-k-median methods on \textit{CensusII} dataset: 8 colors, 100 subsamples of 1000 distinct points each.}
\label{fig:exp-1-k-median-censusII-8-colors}
\end{figure*}
\begin{figure*}[h!]
\subfloat[]{\includegraphics[width=.6125\textwidth]{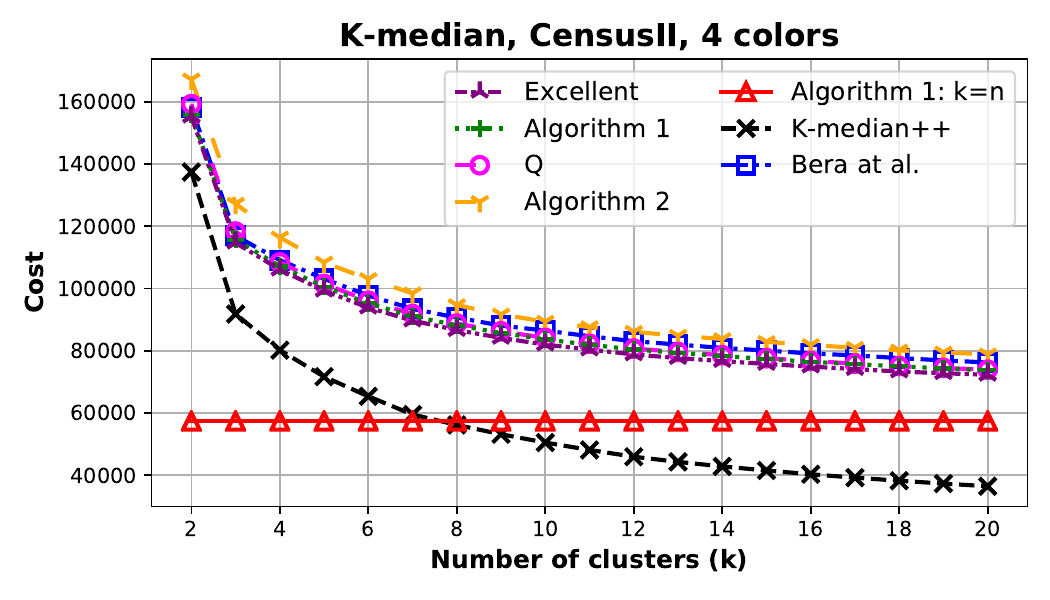}}
\subfloat[]{\includegraphics[width=.35\textwidth]{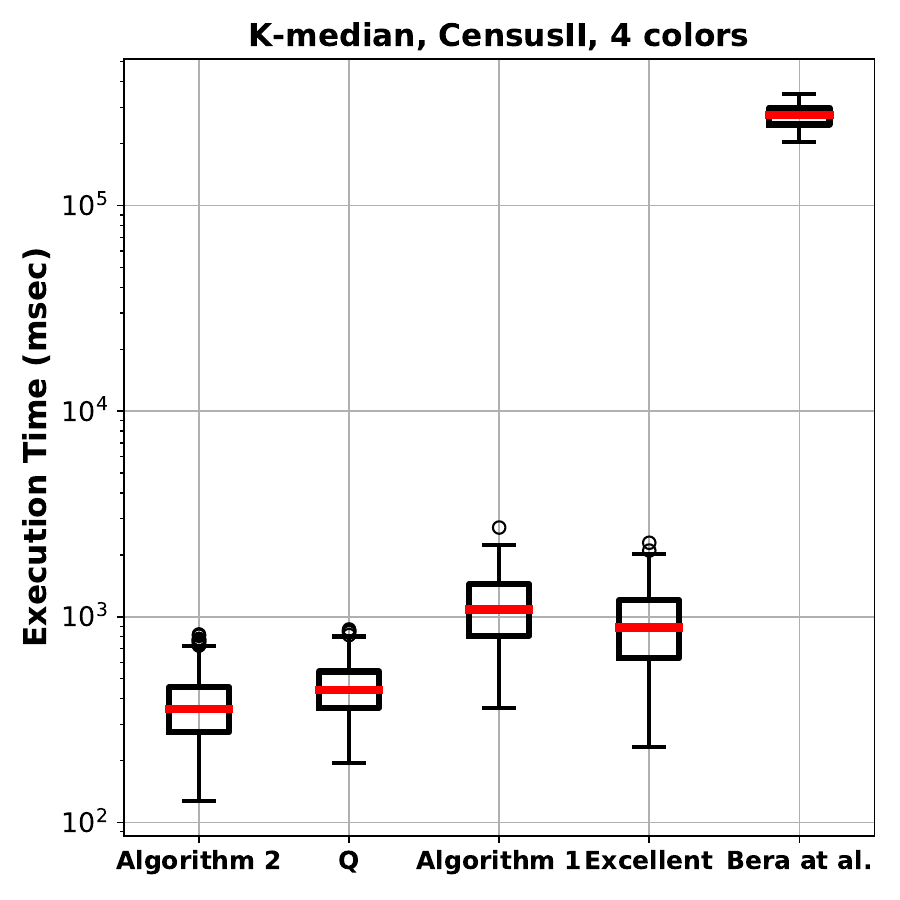}}
\caption{Average cost and execution time of the fair-k-median methods on \textit{CensusII} dataset: 4 colors, 100 subsamples of 1000 distinct points each.}
\label{fig:exp-1-k-median-censusII-4-colors}
\end{figure*}

\begin{table*}[h!]
\begin{center}
\begin{tabular}{|c|c|c|c|}
\hline

            & $k \in [2, 5]$ & $k \in [6, 10]$ & $k \in [11, 20]$ \\ \hline
 {\bf Algorithm~\ref{alg:reduction}: k=n}  & 143647.8 \textit{11160.9} & 143647.8 \textit{11160.9} & 143647.8 \textit{11160.9} \\ \hline
 {\bf k-median++}  & 95258.4 \textit{25778.1} & 56917.8 \textit{5417.5} & 41405.2 \textit{3832.5} \\ \hline
 {\bf Excellent}  & 184695.5 \textit{20692.4} & 160658.8 \textit{10799.8} & 153238.9 \textit{10856.3} \\ \hline
 {\bf Algorithm~\ref{alg:reduction}}  & 187589.5 \textit{20660.7} & 163074.4 \textit{10881.9} & 155139.9 \textit{10832.0} \\ \hline
 {\bf Q}  & 189288.9 \textit{21659.0} & 163719.4 \textit{11099.0} & 155544.5 \textit{10952.7} \\ \hline
 {\bf Algorithm~\ref{alg:randmedian}}  & 205788.6 \textit{33060.7} & 173844.1 \textit{14035.4} & 165075.1 \textit{13902.7} \\ \hline
 {\bf ~\cite{BCN19}}  & 186506.1 \textit{21502.2} & 162025.1 \textit{10575.3} & 154067.1 \textit{10496.4} \\ \hline

\end{tabular}
\caption{Average and standard deviation of the cost of the fair-k-median methods on \textit{CensusII} dataset: 8 colors, 100 subsamples of 1000 distinct points each.}
\label{tab:census-error-8}
\end{center}
\end{table*}

\begin{table*}[h!]
\begin{center}
\begin{tabular}{|c|c|c|c|}
\hline

            & $k \in [2, 5]$ & $k \in [6, 10]$ & $k \in [11, 20]$ \\ \hline
 {\bf Algorithm~\ref{alg:reduction}: k=n}  & 57322.8 \textit{8683.0} & 57322.8 \textit{8683.0} & 57322.8 \textit{8683.0} \\ \hline
 {\bf k-median++}  & 95258.4 \textit{25778.1} & 56917.8 \textit{5417.5} & 41405.2 \textit{3832.5} \\ \hline
 {\bf Excellent}  & 118812.9 \textit{22916.8} & 87243.8 \textit{8554.2} & 75646.0 \textit{8366.8} \\ \hline
 {\bf Algorithm~\ref{alg:reduction}}  & 119795.5 \textit{22706.1} & 88760.4 \textit{8645.0} & 77237.1 \textit{8464.1} \\ \hline
 {\bf Q}  & 121837.5 \textit{23860.1} & 89385.4 \textit{8779.7} & 77455.6 \textit{8511.2} \\ \hline
 {\bf Algorithm~\ref{alg:randmedian}}  & 129809.6 \textit{26354.5} & 95423.7 \textit{11718.4} & 82617.0 \textit{10875.8} \\ \hline
 {\bf ~\cite{BCN19}}  & 121847.6 \textit{22859.6} & 91267.5 \textit{8347.0} & 79891.8 \textit{8120.9} \\ \hline

\end{tabular}
\caption{Average and standard deviation of the cost of the fair-k-median methods on \textit{CensusII} dataset: 4 colors, 100 subsamples of 1000 distinct points each.}
\label{tab:census-error-4}
\end{center}
\end{table*}

%%%%%%%%%%%%%%%%%%%%
%%% Credit cards %%%
\begin{figure*}[h!]
\subfloat[]{\includegraphics[width=.6125\textwidth]{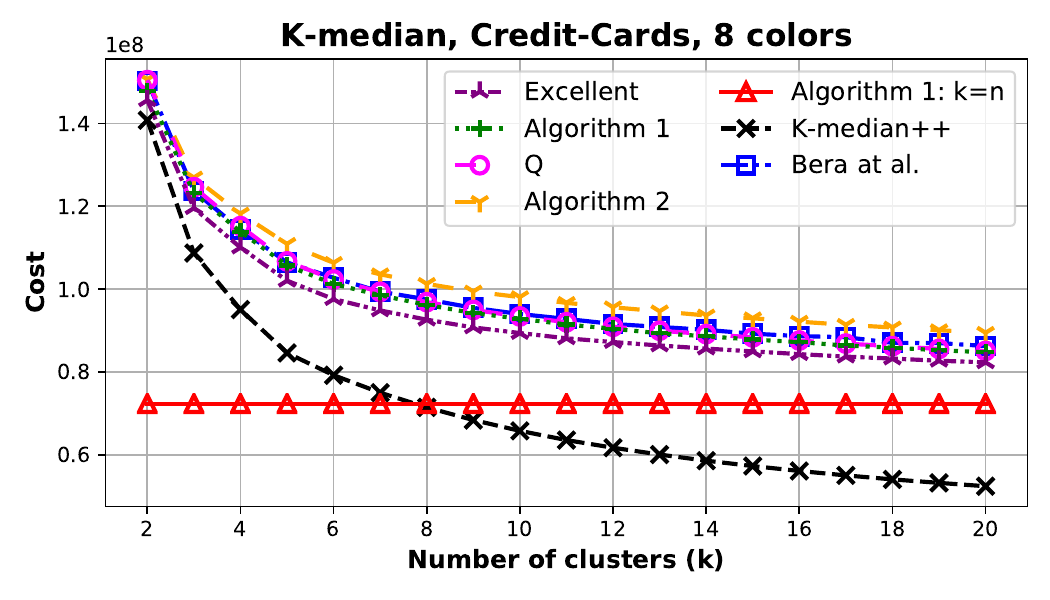}}
\subfloat[]{\includegraphics[width=.35\textwidth]{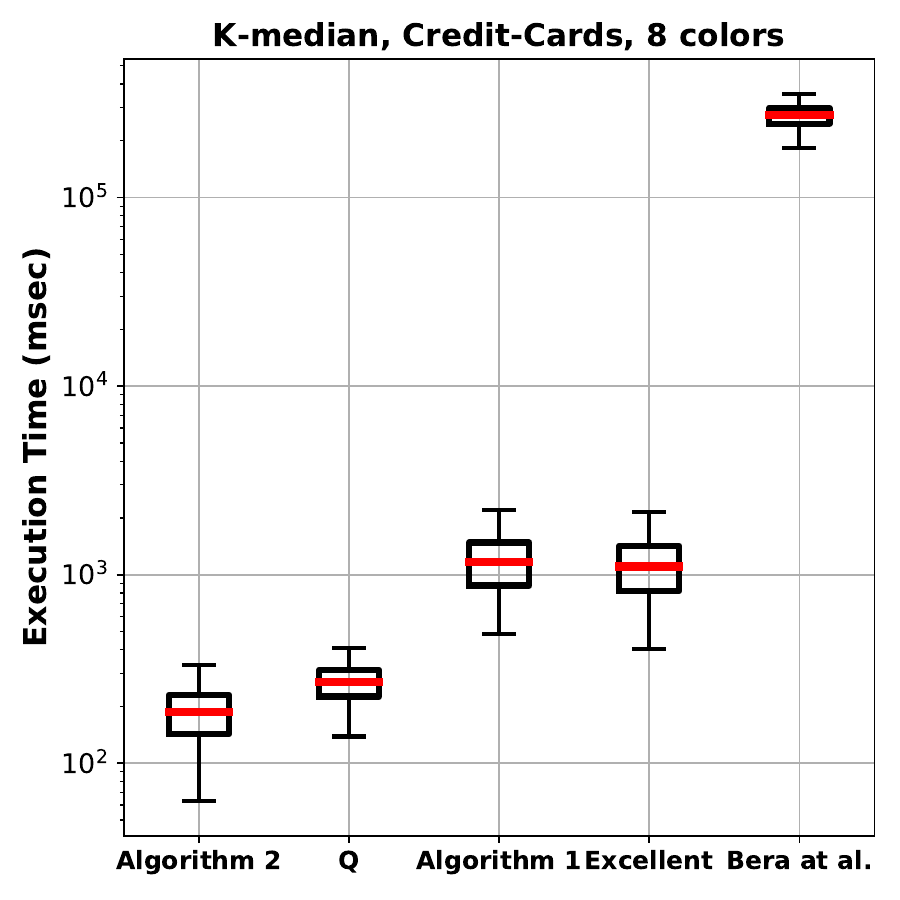}}
\caption{Average cost and execution time of the fair-k-median methods on \textit{Credit cards} dataset: 8 colors, 100 subsamples of 1000 distinct points each.}
\label{fig:exp-1-k-median-credit-8-colors}
\end{figure*}
\begin{figure*}[h!]
\subfloat[]{\includegraphics[width=.6125\textwidth]{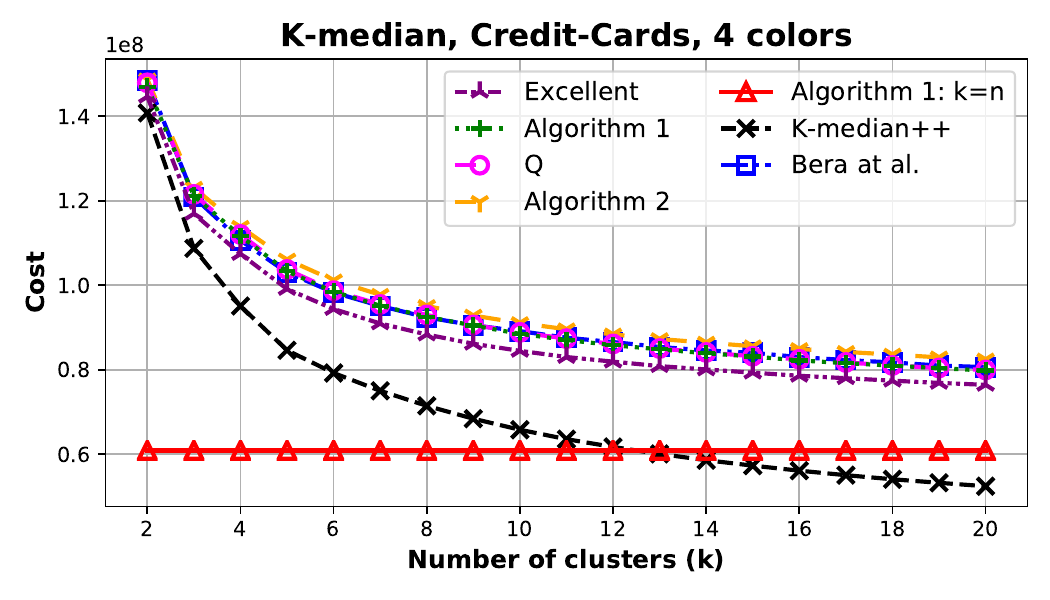}}
\subfloat[]{\includegraphics[width=.35\textwidth]{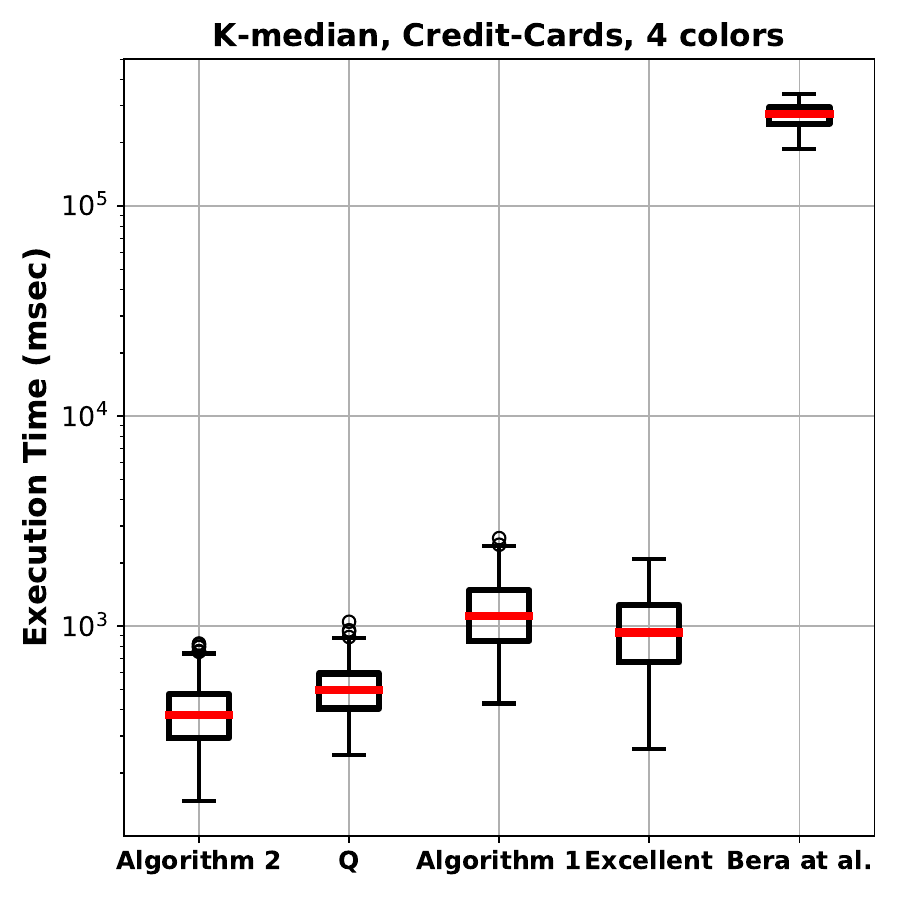}}
\caption{Average cost and execution time of the fair-k-median methods on \textit{Credit cards} dataset: 4 colors, 100 subsamples of 1000 distinct points each.}
\label{fig:exp-1-k-median-credit-4-colors}
\end{figure*}

\begin{table*}[h!]
\begin{center}
\begin{tabular}{|c|c|c|c|}
\hline

            & $k \in [2, 5]$ & $k \in [6, 10]$ & $k \in [11, 20]$ \\ \hline
 {\bf Algorithm~\ref{alg:reduction}: k=n}  & 72339739.1 \textit{3121117.5} & 72339739.1 \textit{3121117.5} & 72339739.1 \textit{3121117.5} \\ \hline
 {\bf k-median++}  & 107302571.2 \textit{21435714.7} & 71975930.7 \textit{5299320.0} & 57207897.4 \textit{4066524.0} \\ \hline
 {\bf Excellent}  & 119312415.8 \textit{16746700.6} & 93053938.9 \textit{4189509.1} & 84871284.9 \textit{3599454.6} \\ \hline
 {\bf Algorithm~\ref{alg:reduction}}  & 122659466.9 \textit{16198184.0} & 96614545.3 \textit{4341943.3} & 87722007.7 \textit{3818697.4} \\ \hline
 {\bf Q}  & 124251499.5 \textit{16883999.8} & 97391340.5 \textit{4602652.1} & 88173991.3 \textit{3940251.9} \\ \hline
 {\bf Algorithm~\ref{alg:randmedian}}  & 126706176.1 \textit{16066799.1} & 101750820.7 \textit{7338227.3} & 92743080.9 \textit{7061486.0} \\ \hline
 {\bf ~\cite{BCN19}}  & 123705922.5 \textit{17036610.1} & 97835628.9 \textit{4640655.8} & 89211162.0 \textit{3889060.7} \\ \hline

\end{tabular}
\caption{Average and standard deviation of the cost of the fair-k-median methods on \textit{Credit cards} dataset: 8 colors, 100 subsamples of 1000 distinct points each.}
\label{tab:cred-error-8}
\end{center}
\end{table*}

\begin{table*}[h!]
\begin{center}
\begin{tabular}{|c|c|c|c|}
\hline

            & $k \in [2, 5]$ & $k \in [6, 10]$ & $k \in [11, 20]$ \\ \hline
 {\bf Algorithm~\ref{alg:reduction}: k=n}  & 60896014.3 \textit{3371582.0} & 60896014.3 \textit{3371582.0} & 60896014.3 \textit{3371582.0} \\ \hline
 {\bf k-median++}  & 107302571.2 \textit{21435714.7} & 71975930.7 \textit{5299320.0} & 57207897.4 \textit{4066524.0} \\ \hline
 {\bf Excellent}  & 117011650.1 \textit{17505918.1} & 88851550.1 \textit{4828231.4} & 79239484.5 \textit{3885547.9} \\ \hline
 {\bf Algorithm~\ref{alg:reduction}}  & 120774309.1 \textit{16788025.7} & 93023954.1 \textit{4925886.7} & 82960786.4 \textit{4133732.4} \\ \hline
 {\bf Q}  & 121293644.5 \textit{17044979.9} & 93318777.1 \textit{4971420.4} & 83129910.2 \textit{4191054.8} \\ \hline
 {\bf Algorithm~\ref{alg:randmedian}}  & 122939880.2 \textit{16637229.6} & 95576130.1 \textit{5432173.1} & 85525305.9 \textit{4762009.5} \\ \hline
 {\bf ~\cite{BCN19}}  & 120760156.9 \textit{17782849.1} & 93085262.4 \textit{4758334.7} & 83669285.3 \textit{4118306.9} \\ \hline

\end{tabular}
\caption{Average and standard deviation of the cost of the fair-k-median methods on \textit{Credit cards} dataset: 4 colors, 100 subsamples of 1000 distinct points each.}
\label{tab:cred-error-4}
\end{center}
\end{table*}

%%%%%%%%%%%%%%%%
%%% DIABETES %%%
\begin{figure*}[h!]
\subfloat[]{\includegraphics[width=.6125\textwidth]{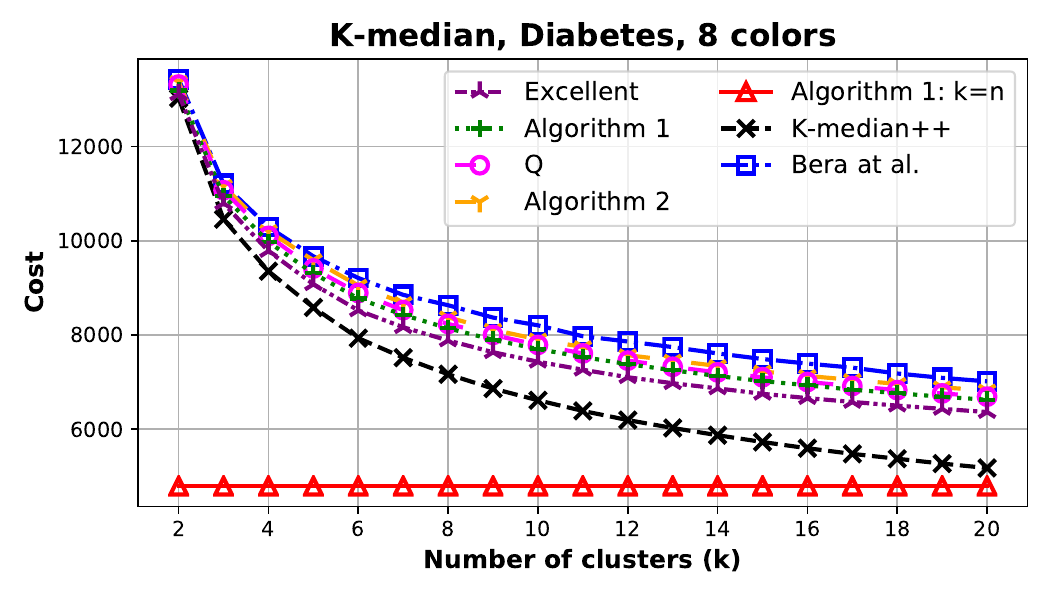}}
\subfloat[]{\includegraphics[width=.35\textwidth]{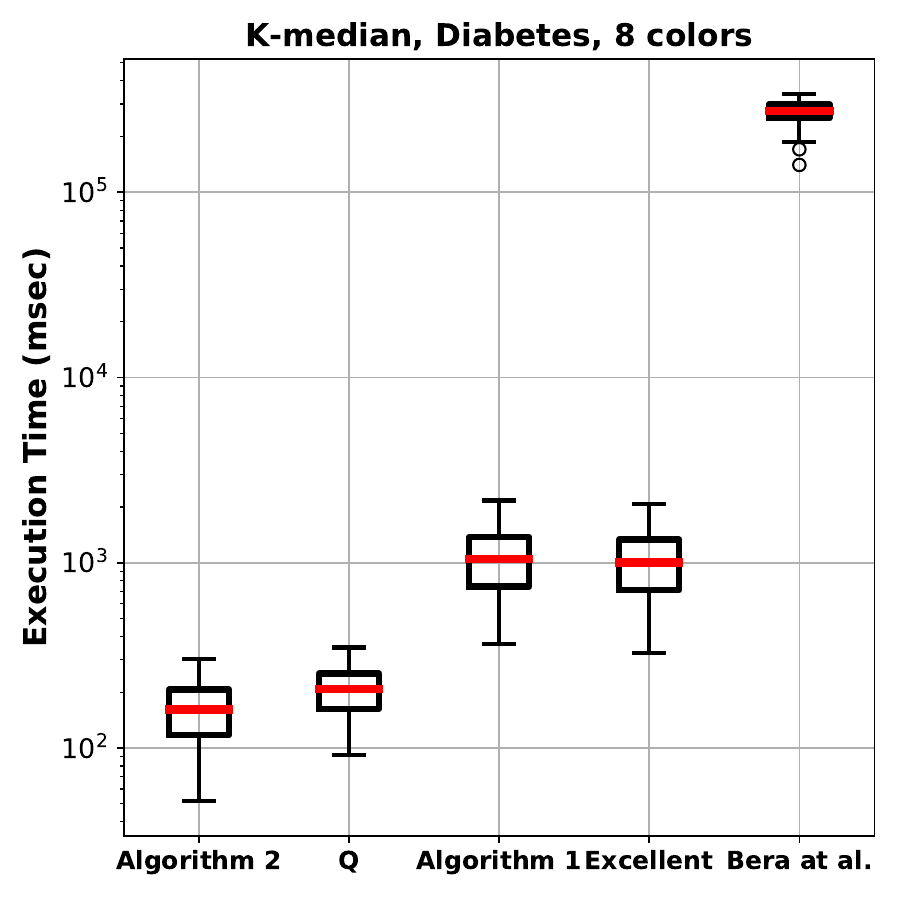}}
\caption{Average cost and execution time of the fair-k-median methods on \textit{Diabetes} dataset: 8 colors, 100 subsamples of 1000 distinct points each.}
\label{fig:exp-1-k-median-diabetes-8-colors}
\end{figure*}
\begin{figure*}[h!]
\subfloat[]{\includegraphics[width=.6125\textwidth]{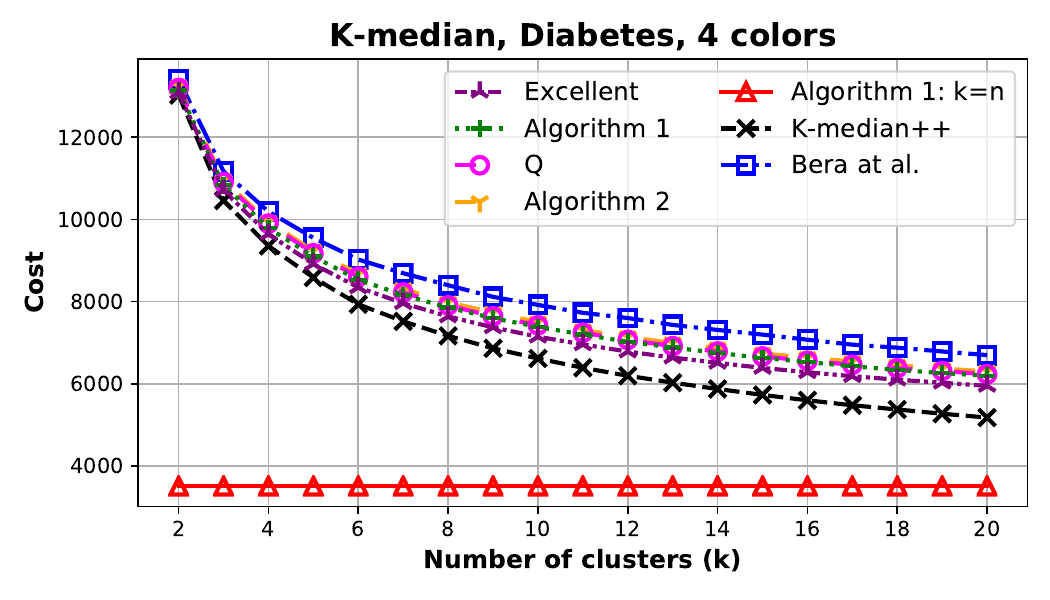}}
\subfloat[]{\includegraphics[width=.35\textwidth]{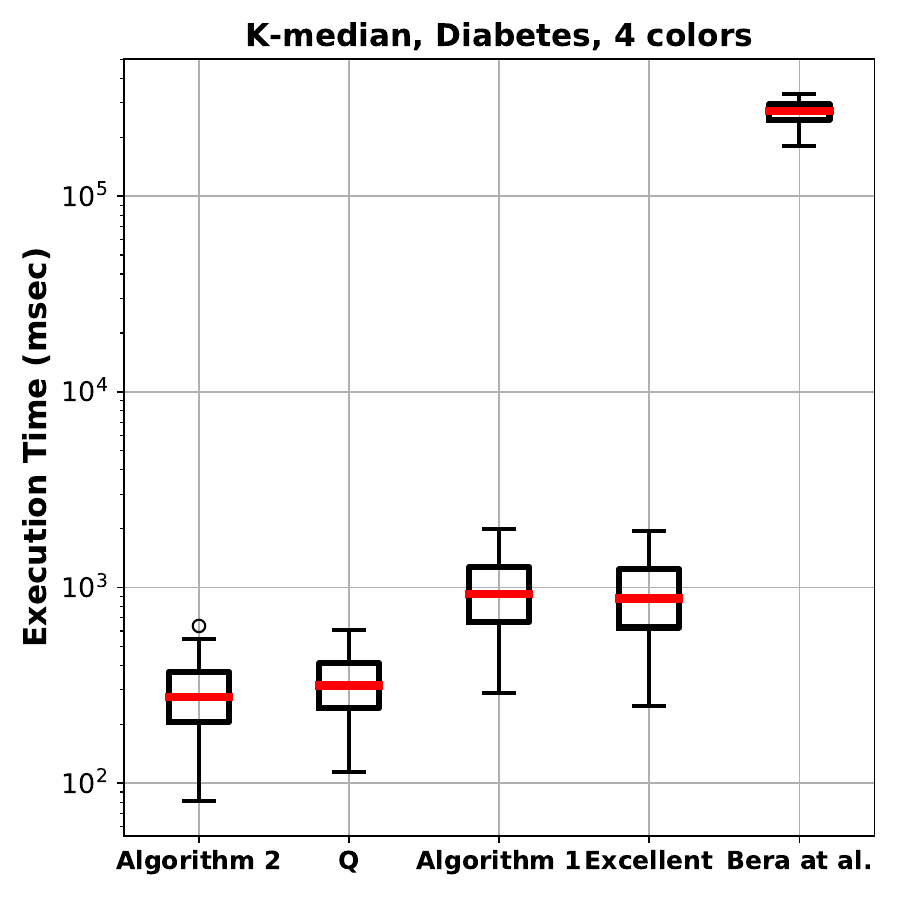}}
\caption{Average cost and execution time of the fair-k-median methods on \textit{Diabetes} dataset: 4 colors, 100 subsamples of 1000 distinct points each.}
\label{fig:exp-1-k-median-diabetes-4-colors}
\end{figure*}

\begin{table*}[h!]
\begin{center}
\begin{tabular}{|c|c|c|c|}
\hline

            & $k \in [2, 5]$ & $k \in [6, 10]$ & $k \in [11, 20]$ \\ \hline
 {\bf Algorithm~\ref{alg:reduction}: k=n}  & 4789.9 \textit{229.5} & 4789.9 \textit{229.5} & 4789.9 \textit{229.5} \\ \hline
 {\bf k-median++}  & 10354.9 \textit{1691.9} & 7218.8 \textit{485.1} & 5710.4 \textit{396.4} \\ \hline
 {\bf Excellent}  & 10689.6 \textit{1541.6} & 7927.3 \textit{431.1} & 6749.5 \textit{343.6} \\ \hline
 {\bf Algorithm~\ref{alg:reduction}}  & 10855.6 \textit{1488.8} & 8188.5 \textit{435.3} & 7016.6 \textit{354.3} \\ \hline
 {\bf Q}  & 10975.4 \textit{1499.6} & 8288.9 \textit{449.5} & 7089.5 \textit{367.7} \\ \hline
 {\bf Algorithm~\ref{alg:randmedian}}  & 11082.9 \textit{1468.3} & 8423.7 \textit{490.0} & 7219.0 \textit{394.0} \\ \hline
 {\bf ~\cite{BCN19}}  & 11148.2 \textit{1454.4} & 8653.1 \textit{444.2} & 7466.4 \textit{392.0} \\ \hline

\end{tabular}
\caption{Average and standard deviation of the cost of the fair-k-median methods on \textit{Diabetes} dataset: 8 colors, 100 subsamples of 1000 distinct points each.}
\label{tab:diab-error-8}
\end{center}
\end{table*}

\begin{table*}[h!]
\begin{center}
\begin{tabular}{|c|c|c|c|}
\hline

            & $k \in [2, 5]$ & $k \in [6, 10]$ & $k \in [11, 20]$ \\ \hline
 {\bf Algorithm~\ref{alg:reduction}: k=n}  & 3505.7 \textit{222.2} & 3505.7 \textit{222.2} & 3505.7 \textit{222.2} \\ \hline
 {\bf k-median++}  & 10354.9 \textit{1691.9} & 7218.8 \textit{485.1} & 5710.4 \textit{396.4} \\ \hline
 {\bf Excellent}  & 10583.8 \textit{1590.7} & 7691.6 \textit{458.9} & 6380.3 \textit{364.3} \\ \hline
 {\bf Algorithm~\ref{alg:reduction}}  & 10713.0 \textit{1546.6} & 7905.0 \textit{446.0} & 6622.5 \textit{368.3} \\ \hline
 {\bf Q}  & 10793.2 \textit{1537.9} & 7958.1 \textit{461.2} & 6661.0 \textit{371.4} \\ \hline
 {\bf Algorithm~\ref{alg:randmedian}}  & 10842.1 \textit{1536.2} & 8039.4 \textit{470.0} & 6732.3 \textit{390.5} \\ \hline
 {\bf ~\cite{BCN19}}  & 11078.4 \textit{1490.2} & 8428.6 \textit{463.1} & 7164.5 \textit{392.5} \\ \hline

\end{tabular}
\caption{Average and standard deviation of the cost of the fair-k-median methods on \textit{Diabetes} dataset: 4 colors, 100 subsamples of 1000 distinct points each.}
\label{tab:diab-error-4}
\end{center}
\end{table*}

%%%%%%%%%%%%%%%%%%%%%
%%% CensusII 450K %%%
\begin{figure*}[h!]
\subfloat[]{\includegraphics[width=.6125\textwidth]{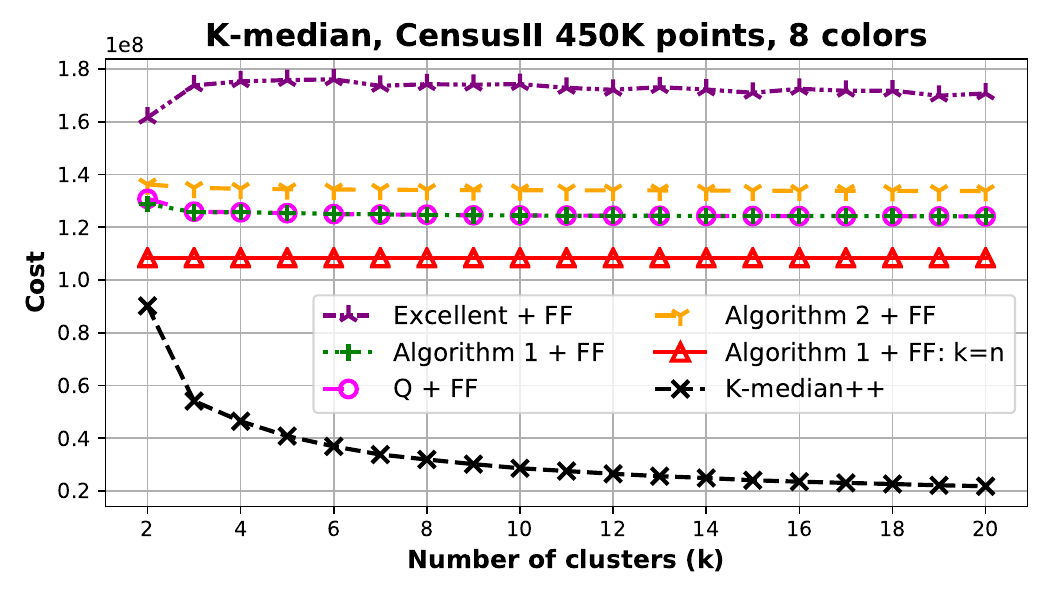}}
\subfloat[]{\includegraphics[width=.35\textwidth]{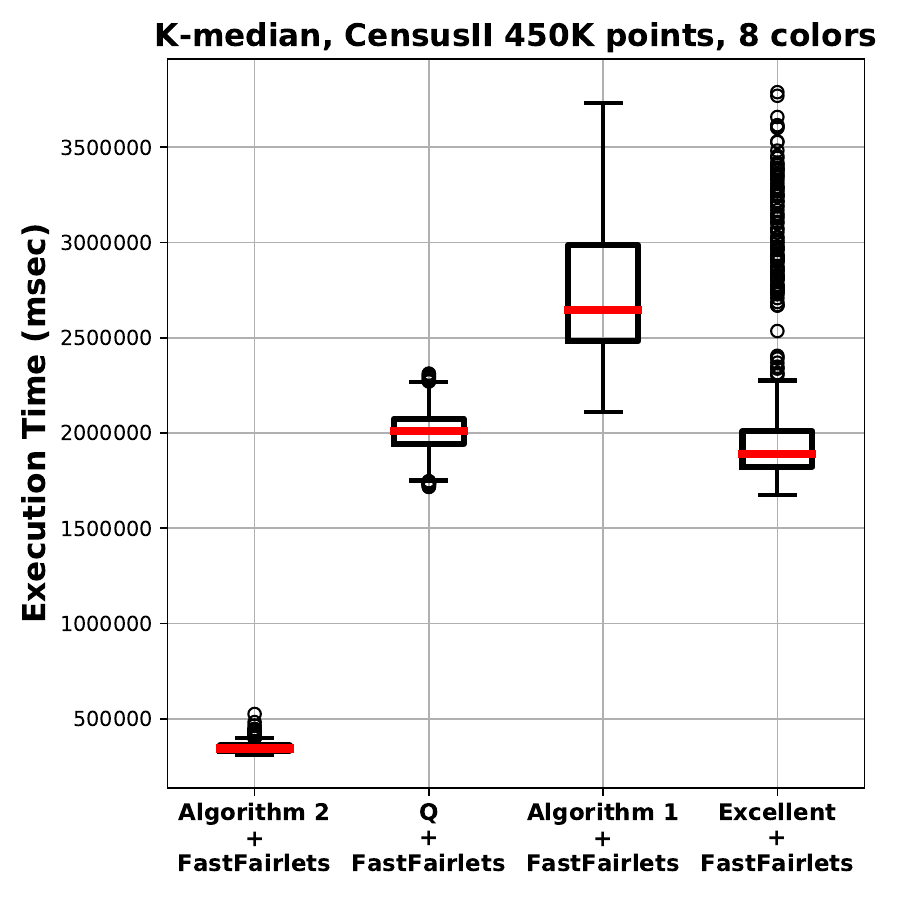}}
\caption{Average cost and execution time of the fair-k-median methods combined with fast fairlets decomposition~\cite{BIOSVW19} on the \textit{CensusII} dataset: 8 colors, 100 subsamples of 450000 distinct points each.}
\label{fig:exp-2-k-median-censusII-8-colors}
\end{figure*}

\begin{table*}[h!]
\begin{center}
\begin{tabular}{|c|c|c|c|}
\hline

            & $k \in [2, 5]$ & $k \in [6, 10]$ & $k \in [11, 20]$ \\ \hline
 {\bf Algorithm~\ref{alg:reduction} + FF: k=n}  & 108369855.9 \textit{3054639.9} & 108369855.9 \textit{3054639.9} & 108369855.9 \textit{3054639.9} \\ \hline
 {\bf k-median++}  & 57863762.6 \textit{19309933.8} & 32275620.7 \textit{3110506.0} & 24170463.4 \textit{1923303.3} \\ \hline
 {\bf Excellent + FF}  & 171581102.9 \textit{8929810.2} & 174418606.9 \textit{6230287.9} & 171759787.4 \textit{5951334.0} \\ \hline
 {\bf Algorithm~\ref{alg:reduction} + FF}  & 126373453.8 \textit{3582880.6} & 124704090.1 \textit{3391359.4} & 124237402.6 \textit{3429534.5} \\ \hline
 {\bf Q + FF}  & 126885878.7 \textit{4115222.4} & 124681188.1 \textit{3454367.0} & 124220881.3 \textit{3424855.2} \\ \hline
 {\bf Algorithm~\ref{alg:randmedian} + FF}  & 135025164.1 \textit{7359895.1} & 134129796.4 \textit{7632188.8} & 133834734.3 \textit{7742148.8} \\ \hline

\end{tabular}
\caption{Average cost and execution time of the fair-k-median methods combined with fast fairlets decomposition~\cite{BIOSVW19} on the \textit{CensusII} dataset: 8 colors, 100 subsamples of 450000 distinct points each.}
\label{fig:exp-2-k-median-censusII}
\end{center}
\end{table*}

For the larger data set 450000-points \textit{CensusII}, we used the fast fairlet decomposition by~\cite{BIOSVW19} to ensure scalability. Unfortunately, the implementation by~\cite{BCN19} could not benefit from this preprocessing step and the implementation itself was not able to process data sets at this scale.
Cost and running time are given in Figure~\ref{fig:exp-2-k-median-censusII-8-colors}.
Running times are similar to those of the small data sets. The most notable difference is that computing an approximate fair assignment after optimization as done by Algorithm {\bf Excellent} negatively affects the approximation.

\section{Conclusion and Future Work}
In this paper, we studied the fair clustering problem in which we are given $n$ points from $\ell$ distinct protected groups and wish to cluster these points such that every group is equally represented in each cluster.
We have presented a generic reduction from fair clustering with multiple (i.e. $\ell\geq 3$) protected classes to to unconstrained clustering at which retains the approximation factor up to constant factors. This result holds for any center-based $k$-clustering objective, including $k$-median, $k$-means, and $k$-center. Moreover, our reduction is robust to approximation and can be easily combined with methods designed to make clustering in general and fair clustering in particular more scalable~\cite{BIOSVW19,HuangJV19,SSS19}.

A number of problems are left in this work. The most challenging one is to show whether there the constant factor loss in the approximation is necessary or not. In other words, does there exist a result showing that fair clustering is strictly harder than unconstrained clustering, for any objective?
Since this question is rather general and might be hard to answer, we propose a few simpler problems. 
First, we have showed that in general metrics, fair $n$-center is APX-hard if the number of colors is greater than $3$. Does this result also hold for the Euclidean plane? Moreover, what can we say about computing a fair $n$-median?
We also showed that a PTAS for fair clustering exists, provided that $k$ is constant. $k$-median and $k$-means in constant dimension admit a PTAS. A natural question is whether such a PTAS also exists for the fair variants of the problem. This problem is open, even in the case of two protected attributes.

Lastly, the balancing constraint we considered in this paper assumes that all protected classes are disjoint and of equal cardinality. If either of these assumptions do not hold, the bi-criteria results by~\cite{BCN19} and~\cite{BGKKRSS19} are still the state of the art. Therefore a further open question is: Under which circumstances is a constant factor approximation to these generalizations obtainable?

\newpage
\bibliographystyle{plain}
\bibliography{references}

\end{document}